\documentclass{article}
\usepackage[margin=1in]{geometry}
\usepackage{natbib}  %
\usepackage{caption} %
\usepackage[colorlinks,citecolor=blue]{hyperref}
\usepackage{microtype}

\usepackage{amsmath}
\usepackage{amssymb}
\usepackage{xcolor}
\usepackage{tikz}
\usepackage{booktabs}

\usepackage{amsmath}
\usepackage{float}
\usepackage{caption}
\usepackage{subcaption}

\usepackage{amsthm}
\usepackage{thm-restate}

\usepackage{mathtools}
\usepackage{array}

\usepackage[capitalize]{cleveref}

\usepackage{algorithm}
\usepackage[noend]{algpseudocode}
\usepackage{comment}
\usepackage{float}
\usepackage{xspace}

\usepackage{physics}
\usepackage{complexity}

\usepackage{multicol}
\usepackage{multirow}
\usepackage{relsize}
\usepackage{makecell}
\usepackage{colortbl}
\usepackage[shortlabels,inline]{enumitem}

\usepackage{autonum} %
\makeatletter
\autonum@generatePatchedReferenceCSL{Cref}
\autonum@generatePatchedReferenceCSL{cref}
\makeatother

\newcommand{\delimit}[3]{\newcommand{#1}[1]{\left#2##1\right#3}}
\delimit \ceil \lceil \rceil
\delimit \floor \lfloor \rfloor

\DeclareMathOperator*{\argmax}{argmax}

\renewcommand{\R}{\mathbb R}

\newcommand{\ie}{{\em i.e.}\xspace}
\newcommand{\eg}{{\em e.g.}\xspace}
\newcommand{\wlo}{WLOG\xspace}
\let\mc\mathcal

\let\eps\varepsilon
\renewcommand\grad\nabla
\let\cite\citep

\let\Root\varnothing
\newcommand{\nature}{\textup{C}}
\newcommand{\gcomm}{\textup{\textsc{Comm}}}
\newcommand{\gcommnone}{\gcomm_{\textsf{\textup{none}}}}
\newcommand{\gcommpriv}{\gcomm_{\textsf{\textup{priv}}}}
\newcommand{\gcommpub}{\gcomm_{\textsf{\textup{pub}}}}
\newcommand{\gsplit}{\textup{\textsc{CSplit}}\xspace}
\newcommand{\gsymsplit}{\textup{\textsc{USplit}}\xspace}
\newcommand{\publicrole}{\textup{\textsc{PublicTeam}}}
\newcommand{\matchingpennies}[1]{\textsf{\textup{MP}}(#1)}
\newcommand{\hrvalue}{\textsf{\textup{CVal}}}
\newcommand{\tmevalue}{\textsf{\textup{TMEVal}}}
\newcommand{\hrvaluepriv}{\hrvalue_{\textsf{\textup{priv}}}}
\newcommand{\symhrvalue}{\textsf{\textup{UVal}}}
\newcommand{\symhrvaluepriv}{\symhrvalue_{\textsf{\textup{priv}}}}

\theoremstyle{plain}
\newtheorem{theorem}{Theorem}
\numberwithin{theorem}{section}

\newtheorem{lemma}[theorem]{Lemma}
\newtheorem{proposition}[theorem]{Proposition}
\newtheorem{corollary}[theorem]{Corollary}

\theoremstyle{definition}
\newtheorem{definition}[theorem]{Definition}

\newcommand{\pmax}{\textcolor{p1color}{\textbf{\textsf{\textup{\smaller MAX}}}}\xspace}
\newcommand{\pmin}{\textcolor{p2color}{\textbf{\textsf{\textup{\smaller MIN}}}}\xspace}
\definecolor{p1color}{RGB}{31,119,180}
\definecolor{p2color}{RGB}{255,127,14}
\definecolor{s2color}{RGB}{44,160,44}
\definecolor{s1color}{RGB}{214,39,40}

\definecolor{mygreen}{rgb}{0.0, 0.5, 0.0}

\let\origthanks\thanks
\renewcommand\thanks[1]{\begingroup\let\rlap\relax\origthanks{#1}\endgroup}

\makeatletter
\def\thm@space@setup{%
  \thm@preskip=\parskip \thm@postskip=0pt
}
\makeatother

\usepackage{parskip}

\title{Hidden-Role Games: Equilibrium Concepts and Computation}

\author{
Luca Carminati\thanks{Equal contribution; author order randomized}\thanks{Work done while at Carnegie Mellon University} \\
Politecnico di Milano\\
\texttt{luca.carminati@polimi.it}\\
\and
Brian Hu Zhang$^*$\\
Carnegie Mellon University\\
\texttt{bhzhang@cs.cmu.edu} \\
\and 
Gabriele Farina \\
MIT\\
\texttt{gfarina@mit.edu} \\
\and 
Nicola Gatti\\
Politecnico di Milano\\
\texttt{nicola.gatti@polimi.it}\\
\and
Tuomas Sandholm \\
Carnegie Mellon University\\
Strategic Machine, Inc. \\
Strategy Robot, Inc. \\
Optimized Markets, Inc. \\
\texttt{sandholm@cs.cmu.edu} \\
}

\date{}

\begin{document}
\maketitle
\begin{abstract}
In this paper, we study the class of games known as \emph{hidden-role games} in which players are assigned {\em privately} to teams and are faced with the challenge of recognizing and cooperating with teammates.
This model includes both popular recreational games such as the {\em Mafia/Werewolf} family and {\em The Resistance (Avalon)} and many real-world settings, such as distributed systems where nodes need to work together to accomplish a goal in the face of possible corruptions.
There has been little to no formal mathematical grounding of such settings in the literature, and it was previously not even clear what the right solution concepts (notions of equilibria) should be.
A suitable notion of equilibrium should take into account the communication channels available to the players (\eg, can they communicate? Can they communicate in private?).  Defining such suitable notions turns out to be a nontrivial task with several surprising consequences. In this paper, we provide the first rigorous definition of equilibrium for hidden-role games, which overcomes serious limitations of other solution concepts not designed for hidden-role games. We then show that in certain cases, including the above recreational games, optimal equilibria can be computed efficiently. 
In most other cases, we show that computing an optimal equilibrium is at least \NP-hard or \coNP-hard. Lastly, we experimentally validate our approach by computing exact equilibria for complete 5- and 6-player \textit{Avalon} instances whose size in terms of number of information sets is larger than $10^{56}$.

\end{abstract}

\newpage
\tableofcontents
\newpage
\section{Introduction}\label{sec:intro}
Consider a multiagent system with communication where the majority of agents share incentives, but there are also hidden defectors who seek to disrupt their progress. 

This paper adopts the lens of game theory to characterize and solve a class of games called \emph{hidden-role games}\footnote{These games are often commonly called {\em social deduction games}.}. Hidden-role games model multi-agent systems in which a team of ``good'' agents work together to achieve some desired goal, but a subset of adversaries hidden among the agents seeks to sabotage the team. Customarily (and in our paper), the ``good'' agents make up a majority of the players, but they will not know who the adversaries are. On the other hand, the adversaries know each other.

Hidden-role games offer a framework for developing optimal strategies in systems and applications that face deception. They have a strong emphasis on communication: players need to communicate in order to establish trust, coordinate actions, exchange information, and distinguish teammates from adversaries.

Hidden-role games can be used to model a wide range of recreational and real-world applications.
Notable recreational examples include the popular tabletop games \emph{Mafia} (also known as \emph{Werewolf}) and \emph{The Resistance}, of which \emph{Avalon} is the best-known variant. 
As an example, consider the game \textit{Mafia}. The players are split in an uninformed majority called \emph{villagers} and an informed minority called \emph{mafiosi}. The game proceeds in two alternating phases, {\em night} and {\em day}. In the night phase, the mafiosi privately communicate and eliminate one of the villagers. In the day phase, players vote to eliminate a suspect through majority voting. 
The game ends when one of the teams is completely eliminated.

We now provide several non-recreational examples of hidden-role games. In many cybersecurity applications~\cite{GarciaMorchon07:Cooperative, GarciaMorchon13:Cooperative, Tripathi22:Integrated}, an adversary compromises and controls some nodes of a distributed system whose functioning depends on cooperation and information sharing among the nodes. The system does not know which nodes have been compromised, and yet it must operate in the presence of the compromised nodes.

Another instance of problems that can be modeled as hidden-role games arises in {\em AI alignment}, \ie, the study of techniques to steer AI systems towards humans' intended goals, preferences, or ethical principles~\cite{Ziegler19:FineTuning,Ji23:AIAlignment,Hubinger24:Sleeper}. 
In this setting, there is risk that a misaligned AI agent may attempt to deceive a human user into trusting its suggestions~\cite{Park23:AIDeception,Scheuer23:Large}.
AI debate \cite{Irving18:AISafety} aims at steering AI agents by using an adversarial training procedure in which a judge has to decide which is the more trustful between two hidden agents, one of which is a deceptor trained to fool the judge.
\citet{Miller21:Chess} proposes an experimental setting consisting of a chess game in which one side is controlled by a player and two advisors, which falls directly under our framework. The advisors pick action suggestions for the player to choose from, but one of the two advisors has the objective of making the team lose.

Hidden-role games also include general scenarios where agents receive inputs from other agents which may be compromised. For example, in {\em federated learning} (a popular category of distributed learning methods), a central server aggregates machine learning models trained by multiple distributed local agents. If some of these agents are compromised, they may send doctored input with the goal of disrupting the training process~\cite{Mothukuri21:Survey}.

Our paper aims to characterize optimal behavior in these settings, and analyze its computability.

\paragraph{Related work.} To the best of our knowledge, there have been no previous works on general hidden-role games. On the other hand, there has been a limited amount of prior work on solving specific hidden-role games.
\citet{Braverman06:Mafia} propose an optimal strategy for \emph{Mafia}, and analyze the win probability when varying the number of players with different roles.
Similarly, \citet{Christiano18:Solving} proposes a theoretical analysis for \emph{Avalon}, investigating the possibility of \emph{whispering}, \ie any two players being able to communicate without being discovered. 
Both of those papers describe game-specific strategies that can be adopted by players to guarantee a specific utility to the teams. %
In contrast, we provide, to our knowledge, the first rigorous definition of a reasonable solution concept for hidden-role games, an  algorithm to find such equilibria, and an experimental evaluation with a wide range of parameterized instances.

Deep reinforcement learning techniques have also been applied to various hidden-role games~\cite{Aitchison21:Learning, Kopparapu22:Hidden, Serrino19:Finding}, but with no theoretical guarantees and usually with no communication allowed between players. 
A more recent stream of works focused on investigating the deceptive capabilities of large language models (LLMs) by having them play a hidden-role game~\cite{Xu23:Language, OGara23:Hoodwinked}. The agents, being LLM-based, communicate using plain human language. However, as before, these are not grounded in any theoretical framework, and indeed we will illustrate that optimal strategies in hidden-role games are likely to involve communication that does not bear resemblance to natural language, such as the execution of cryptographic protocols.

\subsection{Main Modeling Contributions}\label{sec:model-contrib}

We first here give an informal, high-level description of our game model. We also introduce our main solution concept of interest, called {\em hidden-role equilibrium}, and discuss the challenges it addresses. We will define these concepts in more formality beginning in \Cref{sec:eq-concepts}.

We define a (finite) {\em hidden-role game} as an $n$-player finite extensive-form game $\Gamma$ in which the players are partitioned at the start of the game into two teams. Members of the same team share the same utility function, and the game is zero-sum, \ie any gain for one team means a loss for the other. We thus identify the teams as \pmax and \pmin, since teams share the same utility function, but have opposite objectives.
At the start of the game, players are partitioned at random into two teams. %
A crucial assumption is that one of the two teams is \emph{informed}, \ie all the members of that particular team know the team assignment of all the players, while this is not true for all players belonging to the other team. Without loss of generality, we use \pmax to refer to the uninformed team, and \pmin to refer to the informed one.\footnote{For example, in {\em Mafia}, the villagers are \pmax while mafiosi are \pmin.}

To allow our model to cover {\em communication} among players, we formally define the {\em communication extensions} of a game $\Gamma$. The communication extensions are games like $\Gamma$ except that actions allowing messages to be sent between players are explicitly encoded in the game. In the {\em public communication extension}, players are able to publicly broadcast messages. In the {\em private communication extension}, in addition to the public broadcast channel, the players have pairwise private communication channels.\footnote{If players are assumed to be computationally bounded, pairwise private channels can be created from the public broadcast channels through public-key cryptography. However, throughout this paper, for the sake of conceptual cleanliness, we will not assume that players are computationally bounded, and therefore we will distinguish the public-communication case from the private-communication case.} In all cases, communication channels are {\em synchronous} and {\em authenticated}: messages sent on one timestep are received at the next timestep, and are tagged with their sender.
Communication presents the main challenge of hidden role games: \pmax-players wish to share information with teammates, but {\em not} with \pmin-players. 

In defining communication extensions, we must bound the {\em length} of the communication, that is, how many rounds of communication occur in between every move of the game, and how many distinct messages can be sent on each round. To do this, we fix a finite message space\footnote{Note that, if the message space is of size $M$, a message can be sent in $O(\log M)$ bits.} of size $M$ and length of communication $R$, and in our definition of equilibrium we will take a supremum over $M$ and $R$. This will allow us to consider arbitrarily complex message spaces (\ie, $M$ and $R$ arbitrarily large) while still only analyzing finite games: for any {\em fixed} $M$ and $R$, the resulting game is a finite hidden-role game. We will show that our positive results (upper bounds) only require $\log M = R = \polylog(|H|, 1/\eps)$, where $|H|$ is the number of nodes (histories) in the game tree and $\eps$ is the desired precision of equilibrium.

We characterize optimal behavior in the hidden-role setting by converting hidden-role games into team games in a way that preserves the strategic aspect of hidden-roles. This team game is called {\em split-personality form} of a given hidden-role game.
Given a (possibly communication-extended) hidden-role game $\Gamma$, we define and analyze two possible variants:
\begin{itemize}
    \item the \emph{uncoordinated split-personality form} $\gsymsplit(\Gamma)$ is a team games with $2n$ players, derived by splitting each player $i$ in the original game in two distinct players, $i^+$ and $i^-$, that pick actions for $i$ in $\Gamma$ if the player is assigned to team \pmax or \pmin respectively.
    \item the \emph{coordinated split-personality form} $\gsplit(\Gamma)$ is the $(n+1)$-player team game in which the additional player, who we refer to simply as the {\em adversary} or {\em \pmin-player}, takes control of the actions of all players who have been assigned to the \pmin-team. On the contrary, the players from $1$ to $n$ control the players as usual only if they belong to the team \pmax.
\end{itemize}
The coordinated split-personality variant encodes an extra assumption on \pmin's capabilities, namely, that the \pmin-team is controlled by a single player and is therefore perfectly coordinated. Trivially, when only one player is on team \pmin, the uncoordinated and coordinated split-personality forms coincide.

In either case, the resulting game is a team game in which each player has a fixed team assignment. We remark that the split-personality form maintains the strategic aspects of hidden roles, since $i^+$ and $i^-$ share identity when interacting with the environment. For example, players may observe that $i$ has done an action, but do not know if the controller was $i^+$ or $i^-$. Similarly, messages sent by $i^+$ and $i^-$ are signed by player $i$, since the communication extension is applied on $\Gamma$ \emph{before} splitting personalities.

Picking which split-personality variant to use is a modeling assumption that depends on the game instance that one wants to address. For example, in many recreational tabletop games, \gsymsplit is the more reasonable choice because \pmin-players are truly distinct; however, in a network security game where a single adversary controls the corrupted nodes, \gsplit is the more reasonable choice. The choice of which variant also affects the complexity of equilibrium computation: as we will detail in later, \gsplit yields a more tractable solution concept. In certain special cases, however, \gsplit and \gsymsplit will coincide. For example, we will later show that this is the case in {\em Avalon}, which is key to allowing our algorithms to work in that game.

With these pieces in place, we define the {\em hidden-role equilibria} (HRE) of a hidden-role game $\Gamma$ as the {\em team max-min equilibria} (TMEs) of the split-personality form of $\Gamma$. That is, the hidden-role equilibria are the optimal joint strategies for team \pmax in the split-personality game, where optimality is judged by the expected value against a jointly-best-responding \pmin-team. The {\em value} of a hidden-role game is the expected value for \pmax in any hidden-role equilibrium. If communication (private or public) is allowed, we define hidden-role equilibria and values by taking the supremum over $M$ and $R$ of the expected value at the equilibrium, that is, the \pmax-team is allowed to set the parameters of the communication.

Our new solution concept encodes, by design, a pessimistic assumption for the \pmax-team. \pmax picks $M$, $R$ and its strategy considering a worst-case \pmin adversary that knows this strategy and best-responds to it. Throughout our proofs, we will heavily make use of this fact. In particular, we will often consider \pmin-players that ``pretend to be \pmax-players'' under certain circumstances, which is only possible if \pmin-players know \pmax-players' strategies. It is {\em not} allowed for \pmax-players to know \pmin-players' strategies in the same fashion. This is in stark contrast to usual zero-sum game analysis, where various versions of the {\em minimax theorem} promise that the game is unchanged no matter which side commits first to a strategy. Indeed, we discuss in \Cref{sec:duality} the fact that, for hidden-role games, the asymmetry is in some sense necessary: a minimax theorem {\em cannot} hold for nontrivial hidden-role games.
We argue, however, that this asymmetry is natural and {\em inherent} in the the hidden-role setting. %
If we assumed the contrary and inverted the order of the teams so that \pmin commits first to its strategy, \pmax could discover the roles immediately by agreeing to message a passphrase unknown to \pmin in the first round, thus spoiling the whole purpose of hidden-role games. This argument will be made formal in \Cref{sec:duality}.

\paragraph{Existing solution concepts failures}\label{par:eq-concepts-discussion}

We defined our equilibrium notion as a team max-min equilibrium (TME) of the split-personality form of a communication-extended hidden role game. Here, we will argue why some other notions would be less reasonable. 

\begin{itemize}
    \item {\em Nash Equilibrium.} 
    A {\em Nash equilibrium}~\cite{Nash50:Equilibrium} is a strategy profile for all players from which no player can improve its own utility by deviating. This notion is unsuitable for our purposes because it fails to capture {\em team coordination}. In particular, in pure coordination games (in which all players have the same utility function), which are a special case of hidden-role games (with no hidden roles and no adversary team at all), a Nash equilibrium would only be locally optimal in the sense that no player can improve the team's value alone. In contrast, our notion will lead to the optimal team strategy in such games.
    \item {\em Team-correlated equilibrium.} 
    The {\em team max-min equilibrium with correlation}~\cite{vonStengel97:Team,Celli18:Computational} (TMECor), is a common solution concept used in team games. It arises from allowing each team the ability to communicate in private (in particular, to generate correlated randomness) {\em  before} the game begins. For team games, TMECor is arguably a more natural notion than TME, as the corresponding optimization problem is a bilinear saddle-point problem, and therefore in particular the minimax theorem applies, avoiding the issue of which team ought to commit first. However, for hidden-role games, TMECor is undesirable, because it does not make sense for a team to be able to correlate with teammates that have not even been assigned yet. The {\em team max-min equilibrium with communication} (TMECom)~\cite{Celli18:Computational} makes an even stronger assumption about communication among team members, and therefore suffers the same problem.
\end{itemize}

\subsection{Main Computational Contributions}\label{sec:compute-contrib}
We now introduce computational results, both positive and negative, for computing the hidden-role value and hidden-role equilibria of a given game.

\paragraph{Polynomial-time algorithm}
 Our main positive result is summarized in the following informal theorem statement.
\begin{theorem}[Main result, informal; formal result in \Cref{thm:main}]\label{th:main-informal}
    If the number of players is constant, private communication is available, the \pmin-team is a strict minority (\ie, strictly less than half of the players are on the \pmin-team), and the adversary is coordinated, there is a polynomial-time algorithm for exactly computing the hidden-role value.
\end{theorem}
This result should be surprising, for multiple reasons. First, team games are generally hard to solve~\cite{vonStengel97:Team,Zhang22:Team_DAG}, so any positive result for computing equilibria in team games is fairly surprising. Further, it is {\em a priori} not obvious that the value of any hidden-role game with private communication and coordinated adversary is even a {\em rational} number\footnote{assuming all game values and chance probabilities are also rational numbers}, much less computable in polynomial time: for example, there exist adversarial team games with no communication whose TME values are irrational~\cite{vonStengel97:Team}.

There are two key ingredients to the proof of \Cref{th:main-informal}. The first is a special type of game which we call a {\em mediated} game. In a mediated game, there is a player, the {\em mediator}, who is always on team \pmax. \pmax-players can therefore communicate with it and treat it as a trusted party. We show that, when a mediator is present (and all the other assumptions of \Cref{th:main-informal} also hold), the hidden-role value is computable in polynomial time. To do this, we state and prove a form of {\em revelation principle}. Informally, our revelation principle states that, without loss of generality, it suffices to consider \pmax-team strategies in which, at every timestep of the game,
\begin{enumerate}
    \item all \pmax-players send their honest information to the mediator,
    \item the mediator sends action recommendations to all players (regardless of their team allegiance; remember that the mediator may not know the team assignment), and
    \item all \pmax-players play their recommended actions.
\end{enumerate}
\pmin-team players are, of course, free to pretend to be \pmax-team players and thus send false information to the mediator; the mediator must deal with this possibility. However, \pmin-team players cannot just send {\em any} message; they must send messages that {\em are consistent with some \pmax-player}, lest they be immediately revealed as \pmin-team. These observations are sufficient to construct a {\em two-player zero-sum game} $\Gamma_0$, where the mediator is the \pmax-player and the coordinated adversary is the \pmin-player. The value of $\Gamma_0$ is the value of the original hidden-role game, and the size of $\Gamma_0$ is at most polynomially larger than the size of the original game. Since two-player zero-sum extensive-form games can be solved in polynomial time~\cite{Koller94:Fast,Stengel96:Efficient}, it follows that mediated hidden-role games can also be solved in polynomial time.

The second ingredient is to invoke results from the literature on {\em secure multi-party  computation} to {\em simulate} a mediator in the case that one is not already present. %
A well-known result from that literature states that so long as strictly more than half of the players are honest, essentially any interactive protocol---such as the ones used by our mediator to interact with other players---can be simulated efficiently such that the adversary can cause failure of the protocol or leakage of information~\cite{Beaver90:Multiparty,Rabin89:Verifiable}.\footnote{In this part of the argument, the details about the communication channels become important: in particular, the MPC results that we use for our main theorem statement assume that the network is synchronous (\ie, messages sent in round $r$ arrive in round $r+1$), and that there are pairwise private channels and a public broadcast channel that are all authenticated (\ie, message receivers know who sent the message). This is enough to implement MPC so long as $k < n/2$, where $k$ is the number of adversaries and $n$ is the number of players. Our results, however, do not depend on the specific assumptions about the communication channel, so long as said assumptions enable secure MPC with guaranteed outcome delivery.  For a recent survey of MPC, see \citet{Lindell20:Secure}. For example, if $k < n/3$ then MPC does not require a public broadcast channel, so neither do our results. For cleanliness, and to avoid introducing extra formalism, we will stick to one model of communication.} Chaining such a protocol with the argument in the previous paragraph concludes the proof of the main theorem.

\paragraph{Related works on MPC and communication equilibria.}
    The {\em communication equilibrium}~\cite{Forges86:Approach,Myerson86:Multistage} is a notion of equilibrium with a mediator, in which the mediator has two-way communication with all players, and players need to be incentivized to report information honestly and follow recommendations. Communication equilibria include all Nash equilibria, and therefore are unfit for general hidden-role games for the same reason as Nash equilibria, as discussed in the previous subsection. 

    However, when team \pmin has only one player and private communication is allowed, the hidden-role equilibria coincide with the \pmax-team-optimal communication equilibria in the original game $\Gamma$. Our main result covers this case, but an alternative way of computing a hidden-role equilibrium in this special case is to apply the optimal communication equilibrium algorithms of \citet{Zhang22:Polynomial} or \citet{Zhang23:Computing}. However, those algorithms either involve solving linear programs, solving many zero-sum games, or solving zero-sum games with large reward ranges, which will be less efficient than directly solving a single zero-sum game $\Gamma_0$. %

We are not the first to observe that multi-party computation can be used to implement a mediator for use in game theory. In various settings and for various solution concepts, it is known to be possible to implement a mediator using only cheap-talk communication among players~[\eg,~\citealp{Urbano02:Computational,Liu17:Correlation,Abraham06:Distributed,Izmalkov05:Rational}]. For additional reading on the connections between game theory and cryptography, we refer the reader to the survey of \citet{Katz08:Bridging}, and papers citing and cited by this survey. 

\paragraph{Lower bounds.}
We also show {\em lower bounds} on the complexity of computing the hidden-role value, even for a constant number of players, when any of the assumptions in \Cref{th:main-informal} is broken.
\begin{theorem}[Lower bounds, informal; formal statement in \Cref{thm:public-hard,thm:sym-complexity}]\label{thm:sym-complexity-informal} If private communication is disallowed, the hidden-role value problem is \NP-hard. If the \pmin-team is uncoordinated, the problem is \coNP-hard. If both, the problem is $\Sigma_2^\P$-hard. All hardness reductions hold even when the \pmin-team is a minority and the number of players is an absolute constant.
\end{theorem}

\paragraph{Price of hidden roles.} Finally, we define and compute the {\em price of hidden roles}. It is defined (analogously to the price of anarchy and price of stability, which are common quantities of study in game theory) as the ratio between the value of a hidden-role game, and the value of the same game with team assignments made public. We show the following:
\begin{theorem}[Price of hidden roles; formal statement in \Cref{thm:price of hr}]
    Let $D$ be a distribution of team assignments. For the class of games where teams are drawn according to distribution $D$, the price of hidden roles is equal to $1/p$, where $p$ is the probability of the most-likely team in $D$.
\end{theorem}
Intuitively, in the worst case, the \pmax-team can be forced to guess at the beginning of the game all the members of the \pmax-team, and win if and only if its guess is correct. In particular, for the class of $n$-player games with $k$ adversaries, the price of hidden roles is exactly $\binom{n}{k}$. 

\subsection{Experiments: {\em Avalon}} 
We ran experiments on the popular tabletop game {\em The Resistance: Avalon} (or simply {\em Avalon}, for short). As discussed earlier, despite the adversary team in {\em Avalon} not being coordinated in the sense used in the rest of the paper, we show that, at least for the 5- and 6-player variants, the adversary would not benefit from being coordinated; therefore, our polynomial-time algorithms can be used to solve the game. This observation ensures that our main result applies. Game-specific simplifications allow us to reduce the game tree from roughly $10^{56}$ nodes~\cite{Serrino19:Finding} to $10^8$ or even fewer, enabling us to compute exact equilibria. Our experimental evaluation demonstrates the practical efficacy of our techniques on real game instances. Our results are discussed in \Cref{sec:experiments}, and further detail on the game-specific reductions used, as well as a complete hand-analysis of a small {\em Avalon} variant, can be found in \Cref{sec:avalon}.

\subsection{Examples}\label{sec:example-intro}
In this section we present three examples that will hopefully help the reader in understanding our notion of equilibrium and justify some choices we have made in our definition.

\paragraph{A hidden-role matching pennies game.} \label{par:matching-pennies}
Consider a $n$-player version of matching pennies (with $n > 2$), which we denote as $\matchingpennies{n}$. One player is chosen at random to be the adversary (team \pmin). All $n$ players then simultaneously choose bits $b_i \in \{0,1\}$. Team \pmax wins (gets utility $1$) if and only if all $n$ bits match; else, team \pmax loses (gets utility $0$).

With no communication, the value of this game is $1/2^{n-1}$: it is an equilibrium for everyone to play uniformly at random. Public communication does not help, because, conditioned on the public transcript, bits chosen by players must be mutually independent. Thus, the adversary can do the following: pretend to be on team \pmax, wait for all communication to finish, and then play $0$ if the string of all ones is more conditionally likely than the string of all $1$s, and vice-versa.

With private communication, however, the value becomes $1/(n+1)$. Intuitively, the \pmax-team should attempt to guess who the \pmin-player is, and then privately discuss among the remaining $n-1$ players what bit to play. We defer formal proofs of the above game values to \Cref{sec:example}, because they rely on results in \Cref{sec:one-sided}.

\paragraph{Simultaneous actions.} In typical formulations of extensive-form imperfect-information games, it is without loss of generality to assume that games are {\em turn-based}, \ie, only one player acts at any given time. To simulate simultaneous actions with sequential ones, one can simply forbid players from observing each others' actions. However, when communication is allowed arbitrarily throughout the game, the distinction between simultaneous and sequential actions suddenly becomes relevant, because {\em players can communicate when one---but not all---the players have decided on an action}.

To illustrate this, consider the game $\matchingpennies{n}$ defined in the previous section, with public communication, except that the players act in sequence in order of index ($1, 2, \dots, n$). We claim that the value of this game is not $1/2^{n-1}$, but at least $1/2n$. To see this, consider the following strategy for team \pmax. The \pmax players wait for P1 to (privately) pick an action. Then, P2 publicly declares a bit $b \in \{0, 1\}$, and all remaining players play $b$ if they are on team \pmax. If P1 was the \pmin player, this strategy wins with probability at least $1/2$, so the expected value is at least $1/2n$. This example illustrates the importance of allowing simultaneous actions in our game formulations.

\paragraph{Correlated randomness matters.}\label{par:special-player}
We use our third and final example to discuss a nontrivial consequence of the definition of hidden-role equilibrium that may appear strange at first: it is possible for seemingly-useless information to affect strategic decisions and the game value. 

To illustrate, consider the following simple game $\Gamma$: there are three players, and three role cards. Two of the three cards are marked \pmax, and the third is marked \pmin. The cards are dealt privately and randomly to the players. Then, after some communication, all three players simultaneously cast votes to elect a winner. If no player gains a majority of votes, \pmin wins. Otherwise, the elected winner's team wins.
Clearly, \pmax can win no more than $2/3$ of the time in this game: \pmin can simply pretend to be on team \pmax, and in that case \pmax cannot gain information, and the best they can do is elect a random winner. 

Now consider the following seemingly-meaningless modification to the game. We will modify the two \pmax cards so that they are distinguishable. For example, one card has \pmax written on it, and the other has $\pmax'$. We argue that this, perhaps surprisingly, affects the value of the game. In fact, the \pmax team can now win deterministically, even with only public communication. Indeed, consider following strategy. The two players on \pmax publicly declare what is written on their cards (\ie, \pmax or $\pmax'$). The player elected now depends on what the third player did. If one player does not declare \pmax or $\pmax'$, elect either of the other two players. If two players declared \pmax, elect the player who declared $\pmax'$. If two players declare $\pmax'$, elect the player who declared \pmax. This strategy guarantees a win: no matter what the \pmin-player does, any player who makes a unique declaration is guaranteed to be on the \pmax-team. 

What happens in the above example is that making the cards distinguishable introduces a piece of {\em correlated randomness} that \pmax can use: the two \pmax players receive cards whose labels are (perfectly negatively) correlated with each other. Since our definition otherwise prohibits the use of such correlated randomness (because players cannot communicate only with players on a specific team), introducing some into the game can have unintuitive effects. In \Cref{sec:duality}, we expand on the effects of allowing correlated randomness: in particular, we argue that allowing correlated randomness essentially ruins the point of hidden-role games by allowing the \pmax team to learn the entire team assignment. 

\section{Preliminaries}
Our contributions are based on prior work on extensive-form adversarial team games. In this section, we introduce the main definitions related to this class of games.

\begin{definition}
A (perfect-recall, timeable) extensive-form game\footnote{Our definition uses slightly different notation than most works on extensive-form games: for example, we use observations in place of information sets, and we allow for simultaneous moves. In our case, both of these choices will be useful later on in the paper, as we will explicitly make reference to observations and simultaneous actions. We discussed in \Cref{sec:example-intro} why it is important in our model to allow simultaneous actions.} consists of the following.
\begin{enumerate}
\item A set of $n$ players, identified with the integers $[n] = \{1, \dots, n\}$, and an extra player $\nature$ denoting the \emph{chance} player modeling stochasticity.
\item A directed tree of {\em nodes} $H$ (also called {\em histories}). The root of the tree is indicated with $\Root$ while the set of leaves, or {\em terminal nodes}, in $H$ is denoted $Z$.
\item At each node $h \in H \setminus Z$, for each player $i \in [n] \cup \nature$, an {\em action set} $A_i(h)$. The edges leaving $h$ are identified with the joint actions $a \in \bigtimes_{i \in [n] \cup \nature} A_i(h)$.
\item For each player $i \in [n]$, an {\em observation function} $o_i : H \setminus Z \to O$, where $O$ is some set of observations. We assume that the observation uniquely determines the set of legal actions, that is, if $o_i(h) = o_i(h')$ then $A(h) = A(h')$.
\item For each player $i \in [n]$, a {\em utility function} $u_i : Z \to [-1, 1]$. 
\item For each node $h \in H \setminus Z$, a probability distribution $p_h \in \Delta(A_\nature(h))$\, denoting how the chance player picks its action.
\end{enumerate}
\end{definition}

The {\em information state} of player $i$ at node $h \in H$ is the tuple $s_i(h) := (o_i^0 = o_i(\Root), a_i^0, o_i^1, a_i^1, \dots, o_i(h))$ of observations made by player $i$ and actions played by player $i$ on the path to node $h$. We will use $\Sigma_i$ to denote the set of all information states of player $i$. A {\em pure strategy} $x_i$ of player $i$ is a mapping from information states to actions. That is, it is a map $x_i : \Sigma_i \to \cup_{h \in H} A_i(h)$ such that $x_i(s_i(h)) \in A_i(h)$ for all $h$. A {\em mixed strategy} is a probability distribution over pure strategies. A tuple of (possibly mixed) strategies $x = (x_1, \dots, x_n)$ is  a {\em profile} or {\em strategy profile}, and the expected utility of  profile $x$ for player $i$ is denoted $u_i(x)$. We will use $X_i$ to denote the set of pure strategies of player $i$. 

An extensive-form game is an {\em adversarial team game} (ATG) if there is a team assignment $t \in \{ \pmax, \pmin \}^n$ and a team utility function $u : Z \to \R$ such that $u_i(z) = u(z)$ if $t_i = \pmax$, and $u_i(z) = -u(z)$ if $t_i = \pmin$. That is, each player is assigned to a team, all members of the team get the same utility, and the two teams are playing an adversarial zero-sum game\footnote{This is a slight abuse of language: if the \pmax and \pmin teams have different sizes, the sum of all players' utilities is not zero. However, such a game can be made zero-sum by properly scaling each player's utility. The fact that such a rescaling operation does not affect optimal strategies is a basic result for von Neumann--Morgenstern utilities \cite[Chapter 2.4]{Maschler2020:Game}.\label{foot:rescale} We will therefore generally ignore this detail. }.
In this setting, we will write $x_i \in X_i$ and $y_j \in Y_j$ for a generic strategy of a player on team \pmax and \pmin respectively.
ATGs are fairly well studied. In particular, Team Maxmin Equilibria (TMEs) \cite{vonStengel97:Team, Celli18:Computational} and their variants are the common notion of equilibrium employed. The {\em value} of a given strategy profile $x$ for team \pmax is the value that $x$ achieves against a best-responding opponent team. The {\em TME value} is the value of the best strategy profile of team \pmax. That is, the TME value is defined as
\begin{align}
\label{eq:tme}
    \tmevalue(\Gamma) := \max_{x \in \bigtimes_i\! \Delta(X_i)}\ \min_{y \in \bigtimes_j\! \Delta(Y_j)} u(x, y),
\end{align}
and the {\em TMEs} are the strategy profiles $x$ that achieve the maximum value.
Notice that the TME problem is nonconvex, since the objective function $u$ is nonlinear as a function of $x$ and $y$. As such, the minimax theorem does not apply, and swapping the teams may not preserve the solution. Computing an (approximate) TME is $\Sigma_2^\P$-complete in extensive-form games~\cite{Zhang22:Team_DAG}.

\section{Equilibrium Concepts for Hidden-Role Games}\label{sec:eq-concepts}

While the notion of TME is well-suited for ATGs, it is not immediately clear how to generalize it to the setting of hidden-role games. 
We do so by formally defining the concepts of \emph{hidden-role game}, \emph{communication} and \emph{split-personality form} first introduced in \Cref{sec:model-contrib}.

\begin{definition}
An extensive-form game is a {\em zero-sum hidden-role team game}%
, or {\em hidden-role game} for short, if it satisfies the following additional properties:
\begin{enumerate}
\item At the root node, only chance has a nontrivial action set. Chance chooses a string $t \sim \mc D \in \Delta(\{\pmax, \pmin\}^n)$, where $t_i$ denotes the team to which player $i$ has been assigned. Each player $i$ privately observes (at least\footnote{It is allowable for \pmax-players to also have more observability of the team assignment, \eg, certain \pmax-players may know who some \pmin-players are.}) their team assignment $t_i$. In addition, \pmin-players privately observe the entire team assignment $t$.
\item The utility of a player $i$ is defined completely by its team: there is a  $u : Z \to \R$ for which $u_i(z) = u(z)$ if player $i$ is on team \pmax at node $z$, and $-u(z)$ otherwise.\footnote{While at a first look this condition is similar to the one in ATGs, we remark that in this case the number of players in a team depends on the roles assigned at the start. The same considerations as \Cref{foot:rescale} on the zero-sum rescaling of the utilities hold.}
\end{enumerate}
\end{definition}
In some games, players observe additional information beyond just their team assignments. For example, in {\em Avalon}, one \pmax-player is designated {\em Merlin}, and Merlin has additional information compared to other \pmax-players. In such cases, we will distinguish between the {\em team assignment} and {\em role} of a player: the team assignment is just the team that the player is on (\pmax or \pmin), while the role encodes the extra private information of the player as well, which may affect what actions that player is allowed to legally take. For example, the team assignment of the player with role {\em  Merlin} is \pmax. We remark that additional imperfect information of the game may be observed after the root node.%
\footnote{This is an important difference with respect to Bayesian games~\cite{Harsanyi68:Game}, which assume all imperfect information to be the initial \emph{types} of the players. Conversely, we have an imperfect information structure that evolves throughout the game, while only the teams are assigned and observed at the start.}

Throughout this paper, we will use $k$ to denote the largest number of players on the \pmin-team, that is, $k = \max_{t \in \text{supp}(\mc D)}| \{i : t_i = \pmin \}| $.

\subsection{Models of Communication}\label{sec:communication-model}

The bulk of this paper concerns notions of equilibrium that allow communication between the players. We distinguish in this paper between {\em public} and {\em private communication}:
\begin{enumerate}
    \item {\em Public communication}: There is an open broadcast channel on which all players can send messages.
    \item {\em Private communication}: In addition to the open broadcast channel, each pair of players has access to a private communication channel. The private communication channel reveals to all players when messages are sent, but only reveals the message contents to the intended recipients. %
\end{enumerate}

Assuming that public-key cryptography is possible (\eg, assuming the discrete logarithm problem is hard) and players are polynomially computationally bounded, public communication and private communication are equivalent, because players can set up pairwise private channels via public-key exchange. However, in this paper, we assume that agents are computationally unbounded and thus treat the public and private communication cases as different. Our motivation for making this distinction is twofold. First, it is conceptually cleaner to explicitly model private communication, because then our equilibrium notion definitions do not need to reference computational complexity. Second, perhaps counterintuitively, equilibria with public communication only may be {\em more} realistic to execute in practice in human play, precisely {\em because} public-key cryptography breaks. That is, the computationally unbounded adversary renders more ``complex'' strategies of the \pmax-team (involving key exchanges) useless, thus perhaps resulting in a {\em simpler} strategy. 
We emphasize that, in all of our positive results in the paper, the \pmax-team's strategy {\em is} efficiently computable.

To formalize these notions of communication, we now introduce the {\em communication extension}.

\begin{definition}
    The {\em public} and {\em private $(M, R)$-communication extensions} corresponding to a hidden-role game $\Gamma$ are defined as follows. Informally, between every step of the original game $\Gamma$, there will be $R$ rounds of communication; in each round, players can send a public broadcast message and private messages to each player.  The communication extension starts in state $h = \Root \in H_\Gamma$. At each game step of $\Gamma$:
    \begin{enumerate}
        \item Each player $i \in [n]$ observes $o_i(h)$.
        \item For each of $R$ successive communication rounds:
        \begin{enumerate}
            \item Each player $i$ simultaneously chooses a message $m_i \in [M]$ to broadcast publicly.
            \item If private communication is allowed, each player  $i$ also chooses messages $m_{i \to j} \in [M] \cup \{\bot\}$ to send to each player $j \ne i$. $\bot$ denotes that the player does not send a private message at that time.
            \item Each player $j$ observes the messages $m_{i \to j}$ that were sent to it, as well as all messages $m_{i}$ that were sent publicly. That is, by notion of communication, the players observe:
            \begin{itemize}
                \item {\em Public}: player $j$ observes the ordered tuple $(m_1, \dots, m_n)$.
                \item {\em Private}: player $j$ also observes the ordered tuple $(m_{1 \to j}, \dots, m_{n \to j})$, and the set $\{ (i, k) : m_{i \to k} \ne \bot\}$. That is, players observe messages sent to them, and players see when other players send private messages to each other (but not the contents of those messages)    
            \end{itemize}
        \end{enumerate}
        \item Each player, including chance, simultaneously plays an action $a_i \in A_i(h)$. (Chance plays according to its fixed strategy.) The game state $h$ advances accordingly.
    \end{enumerate}
\end{definition}

\noindent We denote the $(M, R)$-extensions as $\gcommpriv^{M, R}(\Gamma)$, and $\gcommpub^{M, R}(\Gamma)$. To unify notation, we also define $\gcommnone^{M, R}(\Gamma) = \Gamma$. When the type of communication allowed and number of rounds are not relevant, we  use $\gcomm(\Gamma)$ to refer to a generic extension.

\subsection{Split Personalities}\label{sec:split-personality}
We introduce two different \emph{split-personality} forms $\gsymsplit(\Gamma)$ and $\gsplit(\Gamma)$ of a hidden-role game $\Gamma$, The split-personality forms are adversarial team games which preserve the characteristics of $\Gamma$.

\begin{definition}\label{def:split}
The {\em uncoordinated split-personality form}\footnote{In the language of Bayesian games, the split-personality form would almost correspond to the {\em agent form}.} of an $n$-player hidden-role game $\Gamma$ is the $2n$-player adversarial team game $\gsymsplit(\Gamma)$ in which each player $i$ is split into two players, $i^+$ and $i^-$, which control player $i$'s actions when $i$ is on team \pmax and team \pmin respectively. 
\end{definition}

Unlike the original hidden-role game $\Gamma$, the split-personality game is an adversarial team game without hidden roles: players $i^+$ are on the \pmax team, and $i^-$ are on the \pmin-team. Therefore, we are able to apply notions of equilibrium for ATGs to $\gsymsplit(\Gamma)$. We also define the {\em coordinated split-personality form}:
\begin{definition}
    The {\em coordinated split-personality form} of an $n$-player hidden-role game $\Gamma$ is the $(n+1)$-player adversarial team game $\gsplit(\Gamma)$ formed by starting with $\gsymsplit(\Gamma)$ and merging all \pmin-players into a single adversary player, who observes all their observations and chooses all their actions.
\end{definition}

Assuming \pmin to be \emph{coordinated} is a worst-case assumption for team \pmax, which however can be justified.
In many common hidden-role games, such as the {\em Mafia} or {\em Werewolf} family of games and most variants of {\em Avalon}, such an assumption is not problematic, because the \pmin-team has essentially perfect information already. In \Cref{sec:avalon:split-personality}, we justify why this assumption is safe also in some more complex {\em Avalon} instances considered.
The coordinated split-personality form will be substantially easier to analyze, and in light of the above equivalence for games like {\em Avalon}, we believe that it is important to study it. %

When team \pmin in $\Gamma$ is already coordinated, that is, if every \pmin-team member has the same observation at every timestep, the coordinated and uncoordinated split-personality games will, for all our purposes, coincide: in this case, any strategy of the adversary in $\gsplit(\Gamma)$ can be matched by a joint strategy of the \pmin-team members in $\gsymsplit(\Gamma)$. This is true in particular if there is only one \pmin-team member. But, we insert here a warning: even when the base game $\Gamma$ has a coordinated adversary team, the private communication extension $\gcommpriv(\Gamma)$ will not. Thus, with private-communication extensions of $\Gamma$, we must distinguish the coordinated and uncoordinated split-personality games even if $\Gamma$ itself is coordinated.

\subsection{Equilibrium Notions}

We now define the notions of equilibrium that we will primarily study in this paper.

\begin{definition}\label{def:hre-value}
 The {\em uncoordinated value} of a hidden-role game $\Gamma$ with notion of communication $c$ is defined as
 \begin{align}\label{eq:limit}
     \symhrvalue_c(\Gamma) := \sup_{M, R} \symhrvalue_c^{M, R}(\Gamma)
 \end{align}
 where $\symhrvalue_c^{M, R}(\Gamma)$ is the TME value of $\gsymsplit(\gcomm_c^{M, R}(\Gamma))$. The {\em coordinated value} $\hrvalue_c(\Gamma)$ is defined analogously by using $\gsplit$.
 \end{definition}
 
 \begin{definition}\label{def:hre-profile}
 An {\em $\eps$-uncoordinated hidden-role equilibrium} of $\Gamma$ with a particular notion of communication $c \in \{ \textsf{none, pub, priv}\}$ is a tuple $(M, R, x)$ where $x$ is a \pmax-strategy profile in $\gsymsplit(\gcomm_c^{M, R}(\Gamma))$ of value at least $\symhrvalue_c(\Gamma) - \eps$. The $\eps$-{\em coordinated hidden-role equilibria} is defined analogously, again with $\gsplit$ and $\hrvalue$ instead of $\gsymsplit$ and $\symhrvalue$.
\end{definition}

As discussed in \Cref{sec:model-contrib}, our notion of equilibrium is inherently asymmetric due to its max-min definition. The \pmax-team is the first to commit to a strategy and a communication scheme, and \pmin is allowed to know both how much communication will be used (\ie, $M$ and $R$) as well as \pmax's entire strategy $x$. As mentioned before, this asymmetry is fundamental in our setting, and we will formalize it in \Cref{sec:duality}.

\section{Computing Hidden-Role Equilibria}
In this section, we show the main computational results regarding the complexity of computing an hidden-role equilibrium in different settings. We first provide positive results for the private-communication case in \Cref{sec:one-sided} while the negative computational results for the no/public-communications cases are presented in \Cref{sec:no/public-comm}. The results are summarized in \Cref{tab:complexity-results}.

\subsection{Computing Private-Communication Equilibria}\label{sec:one-sided}
In this section, we show that it is possible under some assumptions to compute equilibria efficiently for hidden-role games. In particular, in this section, we assume that
\begin{enumerate}
    \item there is private communication,
    \item the adversary is coordinated, and
    \item the adversary is a minority $(k < n/2)$. 
\end{enumerate}

\paragraph{Games with a publicly-known \pmax-player.}\label{sec:mediator}
First, we consider a special class of hidden-role games which we call {\em mediated}. In a mediated game, there is a player, who we call the {\em mediator}, who is always assigned to team \pmax.  The task of the mediator is to coordinate the actions and information transfer of team \pmax. Our main result of this subsection is the following:

\begin{restatable}[Revelation Principle]{theorem}{threvelationprinciple}\label{th:revelation principle}
    Let $\Gamma^*$ be a mediated hidden-role game. Then, for $R \ge 2$ and $M \ge |H|$, there exists a coordinated private-communication equilibrium in which the players on \pmax have a TME profile in which, at every step, the following events happen in sequence:
    \begin{enumerate}
        \item every player on team \pmax sends its observation privately to the mediator,
        \item the mediator sends to every player (\pmax and \pmin) a recommended action, and 
        \item all players on team \pmax play their recommended actions.
    \end{enumerate}
\end{restatable}
 Players on team \pmin, of course, can (and will) lie or deviate from recommendations as they wish. The above revelation principle imples the following algorithmic result:

\begin{restatable}{theorem}{corzerosumefg}\label{cor:zero-sum-efg}
Let $\Gamma^*$ be a mediated hidden-role game, $R \ge 2$, and $M \ge |H|$. An (exact) coordinated private-communication hidden-role equilibrium of $\Gamma^*$ can be computed by solving an extensive-form zero-sum game $\Gamma_0$ with at most $|H|^{k+1}$ nodes, where $H$ is the history set of $\Gamma^*$.
\end{restatable}

Proofs of \Cref{th:revelation principle,cor:zero-sum-efg} are deferred to \Cref{sec:proofs}.

We give a sketch of how the two-player zero-sum game is structured. \Cref{th:revelation principle} allows us to simplify the game by fixing the actions of all players on team \pmax, leaving two strategic players, the mediator and the adversary. Any node from the original game is expanded into three levels:
\begin{enumerate}
    \item the adversary picks messages on behalf of all \pmin-players to send to the mediator,
    \item the mediator picks recommended actions to send to all players, and
    \item the adversary acts on behalf of all \pmin-players.
\end{enumerate}
The key to proving \Cref{cor:zero-sum-efg} is that, in the first step above, the adversary's message space is not too large. Indeed, any message sent by the adversary must be a message that {\em could have plausibly been sent by a \pmax-player}: otherwise the mediator could automatically infer that the sender must be the adversary. It is therefore possible to exclude all other messages from the game since they belong to dominated strategies. Carefully counting the number of such messages would complete the proof.

It is crucial in the above argument that the \pmin-team is coordinated; indeed, otherwise, it would not be valid to model the \pmin-team as a single adversary in $\Gamma_0$. For more elaboration on the case where the \pmin-team is not coordinated, we refer the reader to \Cref{sec:avalon:split-personality}.

In practice, zero-sum extensive-form games can be solved very efficiently in the tabular setting with linear programming~\cite{Koller94:Fast}, or algorithms from the counterfactual regret minimization (CFR) family~\cite{Brown19:Solving,Farina21:Faster,Zinkevich07:Regret}. Thus,  \Cref{cor:zero-sum-efg} gives an efficient algorithm for solving hidden-role games with a mediator.

\paragraph{Simulating mediators with multi-party computation}
In this section, we show that the previous result essentially generalizes (up to exponentially-small error) to games {\em without} a mediator, so long as the \pmin team is also a minority, that is, $k < n/2$. Informally, the main result of this subsection states that, when private communication is allowed, one can efficiently {\em simulate} the existence of a mediator using secure multi-party computation (MPC), and therefore team \pmax can achieve the same value. The form of secure MPC that we use is {\em information-theoretically} secure; that is, it is secure even against computationally-unbounded adversaries.

\begin{theorem}[Main theorem]\label{thm:main}
    Let $\Gamma$ be a hidden-role game with $k < n/2$. Then $\hrvalue_{\textsf{\textup{priv}}}(\Gamma) = \hrvalue_{\textsf{\textup{priv}}}(\Gamma^*)$,
    where $\Gamma^*$ is $\Gamma$ with a mediator added, and moreover this value can be computed in $|H|^{O(k)}$ time by solving a zero-sum game of that size. Moreover, an $\eps$-hidden-role equilibrium with private communication and $\log M = R = \polylog(|H|, 1/\eps)$ can be computed and executed by the \pmax-players in time $\poly(|H|^k, \log(1/\eps))$. 
\end{theorem}

The proof uses MPC to simulate the mediator and then executes the equilibrium given by \Cref{cor:zero-sum-efg}. The proof of \Cref{thm:main}, as well as requisite background on multi-party computation, are deferred to \Cref{sec:mpc}.
We emphasize that \Cref{thm:main,cor:zero-sum-efg} are useful not only for algorithmically computing an equilibrium, but also for manual analysis of games: instead of analyzing the infinite design space of possible messaging protocols, it suffices to analyze the finite zero-sum game $\Gamma_0$. Our experiments on {\em Avalon} use both manual analysis and computational equilibrium finding algorithms to solve instances.

\paragraph{Comparison with communication equilibria.} 
As mentioned in \Cref{sec:compute-contrib}, our construction simulating a mediator bears resemblance to the construction used to define {\em communication equilibria}~\cite{Forges86:Approach,Myerson86:Multistage}. At a high level, a communication equilibrium of a game $\Gamma$ is a Nash equilibrium of $\Gamma$ augmented with a mediator that is playing according to some fixed strategy $\mu$. Indeed, when team \pmin has only one player, it turns out that the two notions coincide:
\begin{restatable}{theorem}{thComm}\label{th:comm}
    Let $\Gamma$ be a hidden-role game with $k =1$. Then $\hrvalue_{\textsf{\textup{priv}}}(\Gamma)$ is exactly the value for \pmax of the \pmax-optimal communication equilibrium of $\Gamma$.
\end{restatable}
However, in the more general case where \pmin can have more than one player, \Cref{th:comm} does not apply: in that case, communication equilibria include all Nash equilibria in particular, and therefore fail to enforce {\em joint} optimality of the \pmin-team, so our concepts and methods are more suitable. The proof is deferred to \Cref{app:comm}.

\subsection{Computing No/Public-Communication Equilibria}\label{sec:no/public-comm}

In this section, we consider games with no communication or with public-communication and a coordinated minority. Conversely to the private-communication case of \Cref{sec:one-sided}, in this case the problem of computing the value of a hidden-role equilibrium is in general \NP-hard.

For the remainder of this section, when discussing the problem of ``computing the value of a game'', we always mean the following promise problem: given a game, a threshold $v$, and an allowable error $\eps > 0$ (both expressed as rational numbers), decide whether the hidden-role value of $\Gamma$ is $\ge v$ or $\le v-\eps$. 

\begin{restatable}{theorem}{publichard}\label{thm:public-hard}
    Even in $2$-vs-$1$ games with public roles and $\eps = 1/\poly(|H|)$, computing the TME value (and hence also the hidden-role value, since adversarial team games are a special case of hidden-role games) with public communication is \NP-hard.
\end{restatable}

Since there is only one \pmin-player in the above reduction, the result applies regardless of whether the adversary is coordinated.

\begin{table}
    \centering
    \definecolor{newcolor}{RGB}{192,255,192}
\scalebox{0.9}{
\def\arraystretch{1.1}
 \begin{tabular}{c|c|c|c}
         & \multicolumn{3}{c}{\bf Communication Type}\\
     \bf Adversary Team \bf Assumptions  & \bf None & \bf Public & \bf Private\\\hline
      \bf  Coordinated, Minority & \multirow{2}{*}{\makecell{$\NP$-complete \\\cite{vonStengel97:Team}}} & \cellcolor{newcolor} \NP-hard & \cellcolor{newcolor} \P~[Thm.~\ref{thm:main}]\ \\\cline{1-1} \cline{4-4}
      \bf  Coordinated &   & \cellcolor{newcolor} [Thm. \ref{thm:public-hard}] & {\em open problem}\\\hline
         \bf Minority & \cellcolor{newcolor} $\Sigma_2^\P$-complete & \cellcolor{newcolor} $\Sigma_2^\P$-hard & \cellcolor{newcolor} \coNP-hard  \\    \cline{1-1}
      \bf None & \cellcolor{newcolor} [Thm. \ref{thm:sym-complexity}] and \cite{Zhang22:Team_DAG} & \cellcolor{newcolor} [Thm. \ref{thm:sym-complexity}] & \cellcolor{newcolor} [Thm. \ref{thm:sym-complexity}] \\
         \hline
    \end{tabular}
}
\captionof{table}{Complexity results for computing hidden-role value with a constant number of players, for various assumptions about the adversary team and notions of communication. The results shaded in green are new to our paper. 
}
\label{tab:complexity-results}
\end{table}

\begin{theorem} \label{thm:sym-complexity}
Even with a constant number of players, a minority adversary team, and $\eps = 1/\poly(|H|)$, computing the uncoordinated value of a hidden-role game is \coNP-hard with private communication and $\Sigma_2^\P$-hard with public communication or no communication.
\end{theorem}
Proofs of results in this section are deferred to \Cref{sec:complexity-proofs}. Intuitively, the proofs work by constructing gadgets that prohibit any useful communication, thus reducing to the case of no communication.

\section{Worked Example}\label{sec:example}
This section includes a worked example of value computation to illustrate the differences among the notions of equilibrium discussed in the paper and illustrates the utility of having a mediator for private communication.
Consider a $n$-player version of matching pennies $\matchingpennies{n}$ as defined in \Cref{par:matching-pennies}.

\begin{restatable}{proposition}{thmmatchingpennies}
\label{thm:matchingpennies}
Let $\matchingpennies{n}$ be the $n$-player matching pennies game.
\begin{enumerate}
    \item The TMECor and TMECom values of $\publicrole(\matchingpennies{n})$ are both $1/2$.
    \item Without communication or with only public communication, the value of $\matchingpennies{n}$ is $1/2^{n-1}$.
    \item With private communication, the value of $\matchingpennies{n}$ is $1/(n+1)$.
\end{enumerate}
\end{restatable}

\begin{proof}
    The first claim, as well as the no-communication value, is known \cite{Basilico17:Team}.
    
    For the public-communication value, observe that, conditioned on the transcript, the bits chosen by the players must be mutually independent of each other. Thus, the adversary can do the following: pretend to be on team \pmax, wait for all communication to finish, and then play $0$ if the string of all ones is more conditionally likely than the string of all $1$s, and vice-versa\footnote{In general, computing the conditional probabilities could take exponential time, but when defining the notion of value here, we are assuming that players have unbounded computational resources. This argument not work for computationally-bounded adversaries. Indeed, if the adversary were computationally bounded, \pmax would be able to use cryptography to build private communication channels and thus implement a mediator, allowing our main positive result \Cref{thm:main} to apply.}. 

    It thus only remains to prove the third claim.

    ({\em Lower bound}) The players simulate a mediator using multi-party computation (see \Cref{thm:main,cor:zero-sum-efg}). Consider the following strategy for the mediator. Sample a string $b \in \{0, 1\}^n$ uniformly at random from the set of $2n+2$ strings that has at most one mismatched bit. Recommend to each player $i$ that they play $b_i$.

    Consider the perspective of the adversary. The adversary sees only a recommended bit $b_i$. Assume WLOG that $b_i = 0$. Then there are $n+1$ possibilities: 
    \begin{enumerate}
        \item $b$ is all zeros ($1$ way)
        \item All other bits of $b$ are $1$ ($1$ way)
        \item Exactly one other bit of $b$ is $1$ ($n-1$ ways).
    \end{enumerate}
    The adversary wins in the third case automatically (since the team has failed to coordinate), and, regardless of what the adversary does, it can win only one of the first two cases. Thus the adversary can win at most $n/(n+1)$ of the time, that is, this strategy achieves value $1/(n+1)$.

    ({\em Upper bound}) Consider the following adversary strategy. The adversary communicates as it would do if it were on team \pmax. Let $b_i$ be the bit that the adversary would play if it were on team \pmax. The adversary plays $b_i$ with probability $1/(n+1)$ and $1-b_i$ otherwise. We need only show that no pure strategy of the medaitor achieves value better than $1/(n+1)$ against this adversary. A strategy of the mediator is identified by a bitstring $b$. If $b$ is all zeros or all ones, the team wins if and only if the adversary plays $b_i$ (probability $1/(n+1)$). If $b$ has a single mismatched bit, the team wins if and only if the mismatched bit is the adversary (probability $1/n$) and the adversary flips $b_i$ (probability $n/(n+1)$). 
\end{proof}

\section{Properties of Hidden-role Equilibria}\label{sec:discussion}
In the following, we discuss interesting properties of hidden-role equilibria given the definition we provided in \Cref{sec:model-contrib}, and that make them fairly unique relative to other notions of equilibrium in team games.

\subsection{The Price of Hidden Roles}\label{sec:price}

One interesting question arising from hidden-role games is the {\em price} of having them. That is, how much value does \pmax lose because roles are hidden? In this section, we define this quantity and derive reasonably tight bounds on it.

\begin{definition}\label{def:pub-team}
The {\em public-team refinement} of an $n$-player hidden-role game $\Gamma$ is the adversarial team game $\publicrole(\Gamma)$ defined by starting with the (uncoordinated) split-personality game, and adding the condition that all team assignments $t_i$ are publicly observed by all players.
\end{definition}

\begin{definition}\label{def:price-hidden-roles}
    For a given hidden-role game $\Gamma$ in which \pmax is guaranteed a nonnegative value (\ie, $u_i(z) \ge 0$ whenever $i$ is on team \pmax), the {\em price of hidden roles} $\textsf{PoHR}(\Gamma)$ is the ratio between the TME value of $\publicrole(\Gamma)$ and the hidden-role value of $\gsymsplit(\Gamma)$.

    For a given class of hidden-role games $\mc G$, the price of hidden roles $\textsf{PoHR}(\mc G)$ is the supremum of the price of hidden roles across all games $\Gamma \in \mc G$.
\end{definition}

Call a hidden-role game {\em normalized} if the value of $\publicrole(\Gamma)$ is $1$, that is, the \pmax-team can win deterministically if the teams are publicly revealed. This assumption is typical in the example games that we consider; for example, it is easy to see that {\em Avalon} satisfies it.

\begin{theorem}\label{thm:price of hr}
    Let $D \in \Delta(\{\pmax, \pmax\}^n)$ be any distribution of teams assignments such that \pmax is always a strict majority. Let $\mathcal G_{n, D}$ be the class of all normalized hidden-role games with $n$ players and team assignment distribution $D$. Then the price of hidden roles of $\mathcal G_{n, D}$ is exactly the largest probability assigned to any team by $D$, that is,
    \begin{align}
        \mathsf{PoHR}(\mathcal G_{n, D}) = \max_{t \in \{\pmax, \pmin\}^n} \Pr_{t' \sim D} [t' = t].
    \end{align}
    The lower bound is achieved even in the presence of private communication.
\end{theorem}

\begin{proof}[Proof]
Let $t^*$ be the team to which $D$ assigns the highest probability, and let $p^*$ be that probability. Our goal is to show that the price of hidden roles is $1/p^*$.

({\em Upper bound}) Team \pmax assumes that the true \pmax-team is exactly the team $t^*$. Then \pmax gets  utility at most a factor of $1/p^*$ worse than the TME value of $\publicrole(\Gamma)$: if the assumption is correct, then \pmax gets the TME value; if the assumption is incorrect, \pmax gets value at least $0$ thanks to the condition on \pmax's utilities in \Cref{def:price-hidden-roles}.

 ({\em Lower bound}) Consider the following game $\Gamma$. Nature first selects a team assignment $t \sim D$ and each player privately observes its team assignment. Then, all players are simultaneously asked to announce what they believe the true team assignment is. The \pmax-team wins if every \pmax-player announces the true team assignment. If \pmax wins, \pmax gets utility $1$; otherwise \pmax gets utility $0$.

 Clearly, if teams are made public, \pmax wins easily. With teams not public, suppose that we add a mediator to the game so that \Cref{th:revelation principle} applies. This cannot decrease \pmax's value. The mediator's strategy amounts to selecting what team each player should announce. Mediator strategies in which different players announce different teams are dominated. The mediator strategy in which the mediator tells every player to announce team $t$ wins if and only if $t$ is the true team, which happens with probability at most $p^*$ (if $t = t^*$). Thus, even the game with a mediator added has value at most $p^*$, completing the proof.
\end{proof}
This implies immediately:
\begin{corollary}
    Let $\mathcal G_{n, k}$ be the class of all normalized hidden-role games where the number of players and adversaries are always exactly $n$ and $k$ respectively, with $k < n/2$. The price of hidden roles in $\mathcal G_{n, k}$ is exactly $\binom{n}{k}$.
\end{corollary}
In particular, when $k = 1$, the price of hidden roles is at worst $n$. This is in sharp contrast to the {\em price of communication} and {\em price of correlation} in ATGs, both of which can be arbitrarily large even when $n=3$ and $k=1$~\cite{Basilico17:Team,Celli18:Computational}.

\subsection{Order of Commitment and Duality Gap}\label{sec:duality}

In \Cref{def:hre-value}, when choosing the TME as our solution concept and defining the split-personality game, we explicitly choose that \pmax should pick its strategy before \pmin---that is, the team committing to a strategy is the same one that has incomplete information about the roles. One may ask whether this choice is necessary or relevant: for example, what happens when the TME problem \eqref{eq:tme} satisfies the minimax theorem? Perhaps surprisingly, the answer to this question is that, at least with private communication, {\em the minimax theorem in hidden-role games only holds in ``trivial'' cases}, in particular, when the hidden-role game is equivalent to its public-role refinement (\Cref{def:pub-team}).

\begin{restatable}{proposition}{prsaddlepointzero}\label{pr:saddle-point-zero}
    Let $\Gamma$ be any hidden-role game. Define $\symhrvaluepriv'(\Gamma)$ identically to $\symhrvaluepriv(\Gamma)$, except that \pmin commits before \pmax---that is, in \eqref{eq:tme}, the maximization and minimization are flipped. Then $\symhrvaluepriv'(\Gamma)
    $ is equal to the TME value of $\publicrole(\Gamma)$ with communication---that is, the equilibrium value of the zero-sum game in which teams are public and intra-team communication is private and unlimited.
\end{restatable}

\begin{proof}
It suffices to show that team \pmax can always cause the teams to be revealed publicly if \pmin commits first. Let $s$ be a long random string. All members of team \pmax broadcast $s$ publicly at the start of the game. Since \pmin commits first, \pmin cannot know or guess $s$ if it is sufficiently long; thus, with exponentially-good probability, this completely reveals the teams publicly. Then, using the private communication channels, team \pmax can play a TMECom of $\publicrole(\Gamma)$.
\end{proof}

Therefore, the choice of having \pmax commit to a strategy before \pmin is forced upon us: flipping the order of commitment would ruin the point of hidden-role games.

\section{Experimental Evaluation: {\em Avalon}}\label{sec:experiments}

\begin{table*}[t]
\setlength\tabcolsep{3mm}
\centering
 \scalebox{1}{
    \begin{tabular}{lrr}
        \bf Variant & \multicolumn{1}{r}{\bf 5 Players} & \multicolumn{1}{r}{\bf 6 Players}
    \\\toprule
       No special roles ({\em Resistance}) & 3 / 10 $=$ 0.3000\rlap{*} & 1 / 3 $\approx$ 0.3333\rlap{*}  \\
       Merlin   & 2 / 3 $\approx$ 0.6667\rlap{*} & 3 / 4 $=$ 0.7500\rlap{*} \\
       Merlin \& Mordred & 731 / 1782 $\approx$ 0.4102 & 6543 / 12464 $\approx$ 0.5250 \\
       Merlin \& 2 Mordreds & 5 / 18 $\approx$ 0.2778 & 99 / 340 $\approx$ 0.2912 \\
       Merlin, Mordred, Percival, Morgana & 67 / 120 $\approx$ 0.5583 & --- \\\bottomrule
    \end{tabular}}
    \captionof{table}{Exact equilibrium values  for 5- and 6-player {\em Avalon}. The values marked * were also manually derived by \citet{Christiano18:Solving}; we match their results. `---': too large to solve.}\label{tab:experiments}
\end{table*}

In this section, we apply \Cref{cor:zero-sum-efg} to instances of the popular hidden-role game {\em The Resistance: Avalon} (hereafter simply {\em Avalon}). We solve various versions of the game with up to six players. 

A game of {\em Avalon} proceeds, generically speaking, as follows. There are $n$ players, $\ceil{n/3}$ of which are randomly assigned to team \pmin and the rest to team \pmax. Team \pmin is informed. Some special rules allow players observe further information; for example, {\em Merlin} is a \pmax-player who observes the identity of the players on team \pmin, except the \pmin-player {\em Mordred}, and the \pmax-player {\em Percival} knows {\em Merlin} and {\em Morgana} (who is on team \pmin), but does not know which is which. Players are assigned roles by means of privately-dealt role cards. The role cards are treated as indistinguishable.\footnote{In a typical {\em Avalon} card deck, the \pmax ``generic'' (\ie, non-Merlin, non-Percival) role cards {\em are} actually distinguishable. As we have discussed in \Cref{sec:example-intro}, this can have an impact on optimal play. However, it is a strong convention at all levels of human play that players do not exploit the distinguishability of role cards. We therefore generic \pmax role cards as indistinguishable.}

The game proceeds in five rounds. In each round, a {\em leader} publicly selects a certain number of people (defined as a function of the number of players and current round number) to go on a {\em mission}. Players then publicly vote on whether to accept the leader's choice. If a strict majority vote to accept, the mission begins; otherwise, leadership goes to the player to the left. If four votes pass with no mission selected, there is no vote on the fifth mission (it automatically gets accepted). If a \pmin-player is sent on a mission, they have the chance to {\em fail} the mission.  The goal of \pmax is to have three missions pass. If {\em Merlin} is present, \pmin also wins by correctly guessing the identity of Merlin at the end of the game. {\em Avalon} is therefore parameterized by the number of players and the presence of the extra roles {\em Merlin, Mordred, Percival}, and {\em Morgana}.

{\em Avalon} is far too large to be written in memory: \citet{Serrino19:Finding} calculates that $5$-player {\em Avalon} has at least $10^{56}$ information sets. However, in {\em Avalon} with $\le 6$ players, many simplifications can be made to the zero-sum game given by \Cref{cor:zero-sum-efg} without changing the equilibrium. These are detailed in \Cref{sec:avalon}, but here we sketch one of them which has theoretical implications. Without loss of generality, in the zero-sum game in \Cref{cor:zero-sum-efg}, the mediator completely dictates the choice of missions by telling everyone to propose the same mission and vote to accept missions, and \pmin can do nothing to stop this. Therefore, team \pmin always has symmetric information in the game: they know each others' roles (at least when $n \le 6$), and the mediator's recommendations to the players may as well be public. Therefore, {\em Avalon} is already natively without loss of generality a game with a coordinated adversary in the sense of \Cref{sec:split-personality}, so the seemingly strong assumptions used in \Cref{def:split} are in fact appropriate in {\em Avalon}. Even after our simplifications, the games are fairly large, \emph{e.g.}, the largest instance we solve has 2.2 million infosets and 26 million terminal nodes.%

Our results are summarized in \Cref{tab:experiments}. Games were solved using a CPU compute cluster machine with 64 CPUs and 480 GB RAM, using two algorithms:
\begin{enumerate}
    \item A parallelized version of the PCFR+ algorithm~\cite{Farina21:Faster}, a scalable no-regret learning algorithm. PCFR+ was able to find an approximate equilibrium with exploitability $<10^{-3}$ in less than 10 minutes in the largest game instance, and was able to complete 10,000 iterations in under two hours for each game.
    \item An implementation of the simplex algorithm with exact (rational) precision, which was warmstarted using incrementally higher-precision solutions obtained from configurable finite-precision floating-point arithmetic implementation of the simplex algorithm, using an algorithm similar to that of~\citet{Farina18:Practical}. This method incurred significantly higher runtimes (in the order of hours to tens of hours), but had the advantage of computing \emph{exact} game values at equilibrium. 
\end{enumerate}
\Cref{tab:experiments} shows exact game values for the instances we solved.

\paragraph{Findings} 
We solve {\em Avalon} exactly in several instances with up to six players. In the simplest instances ({\em Resistance} or only Merlin), \citet{Christiano18:Solving} previously computed equilibrium values by hand. The fact that we match those results is positive evidence of the soundness of both our equilibrium concepts and our algorithms.

Curiously, as seen in \Cref{tab:experiments}, the game values are not ``nice'' fractions: this suggests to us that most of the equilibrium strategies will likely be inscrutable to humans. The simplest equilibrium not previously noted by Christiano, namely Merlin + 2 Mordreds with 5 players, is scrutable, and is analyzed in detail in \Cref{sec:appendix-5p2m-eqm}.

Also curiously, having Merlin and two Mordreds (\ie, having a Merlin that does not actually know anything) is not the same as having no Merlin. If it were, we would expect the value of Merlin and two Mordreds to be $0.3 \times 2/3 = 0.2$ (due to the 1/3 probability of \pmin randomly guessing Merlin). But, the value is actually closer to $0.28$. The discrepancy is due to the ``special player'' implicit correlation discussed in \Cref{sec:example-intro}.

\section{Conclusions and Future Research}
In this paper, we have initiated the formal study of hidden-role games from a game-theoretic perspective. We build on the growing literature on ATGs to define a notion of equilibrium, and give both positive and negative results surrounding the efficient computation of these equilibria. In experiments, we completely solve real-world instances of {\em Avalon}. As this paper introduces a new and interesting class of games, we hope that it will be the basis of many future papers as well. We leave many interesting questions open.
\begin{enumerate}
    \item From our results, it is not even clear that hidden-role equilibria and values can be computed in {\em finite} time except as given by \Cref{thm:main}. Is this possible? For example, is there a revelation-principle-like characterization for {\em public} communication that would allow us to fix the structure of the communication? We believe this question is particularly important, as humans playing hidden-role games are often restricted to communicating in public and cannot reasonably run the cryptographic protocols necessary to build private communication channels or perform secure MPC.
    \item Changing the way in which communication works can have a ripple effect on the whole paper. One particular interesting change that we do not investigate is {\em anonymous} messaging, in which players can, publicly or privately, send messages that do not leak their own identity. How does the possibility of anonymous messaging affect the central results of this paper?
    \item In this paper, we do not investigate or define hidden-role games where {\em both} teams have imperfect information about the team assignment. What difference would that make? In particular, is there a way to define an equilibrium concept in that setting that is ``symmetric'' in the sense that it does not require a seemingly-arbitrary choice of which team ought to commit first to its strategy? 
\end{enumerate}

\section*{Acknowledgements}
The work of Prof. Sandholm’s group is supported by the National Science Foundation under grants RI2312342 and RI-1901403, and by the ARO under award W911NF2210266. Nicola Gatti is supported by
FAIR (Future Artificial Intelligence Research) project, funded by the NextGenerationEU program within
the PNRR-PE-AI scheme (M4C2, Investment 1.3, Line on Artificial Intelligence). We thank Justin Raizes
for useful discussions on multi-party computation, and Cristian Palma Foster for pointing out an error in an earlier version of the statement of \Cref{thm:price of hr}.
\clearpage

\bibliographystyle{plainnat}
\bibliography{dairefs}
\appendix

\setlist{itemsep=\parskip,parsep=\parskip}
\newpage
\onecolumn

\section{Mediated Games and the Revelation Principle}\label{sec:proofs}
In this section, we prove \Cref{th:revelation principle} and \Cref{cor:zero-sum-efg}.

\subsection{\Cref{th:revelation principle}}

\threvelationprinciple*
\begin{proof}
    We follow the usual proof structure of revelation principle proofs. Let $x = (x_1, \dots, x_n)$ be any strategy profile for team \pmax in $\gsplit(\gcommpriv(\Gamma^*))$.
    Consider the strategy profile $x' = (x_1', \dots, x_n')$ that operates as follows. For each player $i$, the mediator instantiates a simulated version of each player $i$ playing according to strategy $x_i$. These simulated players are entirely ``within the imagination of the mediator''. 
    \begin{enumerate}
        \item When a (real) player $i$ sends an observation $o_i(h)$ to the mediator, the mediator forwards observation $o_i(h)$ to the simulated player $i$.
        \item When a simulated player $i$ wants to send a message to another player $j$, the mediator forwards the message to the {\em simulated} player $j$. 
        \item When a simulated player $i$ plays an action $a_i$, the mediator forwards the action as a message to the real player $i$.
    \end{enumerate}

    Since the strategy $x_i$ is only well-defined on sequences that can actually arise in $\gsplit(\gcommpriv(\Gamma^*))$, the simulated player $i$ may crash if it receives an impossible sequence of observations. If player $i$'s simulator has crashed, it will no longer send simulated messages, and the mediator will no longer send messages to player $i$.
    
    It suffices to show that team \pmin cannot exploit $x'$ more than $x$. Let $y'$ be any best-response strategy profile for \pmin against $x'$. %
    
    We will show that there exists a strategy $y$ such that $(x', y')$  is equivalent to $(x, y)$. Consider the strategy $y$ for team \pmin in which each player $i$ maintains simulators of both $x_i$ and $y'_i$, and acts as follows.
    \begin{enumerate}
        \item Upon receiving an observation or message, forward it to $y_i'$
        \item If $y_i'$ wants to send an observation to the mediator, forward that observation to $x_i$. 
        \item If $x_i$ sends a message, send that message.
        \item If $x_i$ plays an action $a_i$, forward that action to $y_i'$ as a message from the mediator. If $x_i$ crashes, send empty messages to $y_i'$ from the mediator. In either case, when $y_i'$ outputs an action, play that action.
    \end{enumerate}
    By definition, the profiles $(x, y)$ and $(x', y')$ have the same expected utility (in fact, they are equivalent, in the sense that they induce the same outcome distribution over the terminal nodes of $\Gamma$), so we are done.
\end{proof}

\subsection{\Cref{cor:zero-sum-efg}}
\corzerosumefg*

\begin{proof}
    Consider the zero-sum game $\Gamma_0$ that works as follows. There are two players: the {\em mediator}, representing team \pmax, and a single {\em adversary}, representing team \pmin.  The game state of $\Gamma_0$ consists of a history $h$ in $\Gamma$ (initially the root), and $n$ {\em sequences} $s_i$. Define a sequence $s_i$ to be {\em consistent} if it is a prefix of a terminal sequence $s_i(z)$ for some terminal node $z$ of $\Gamma$. For each sequence $s_i$ which ends with an action, let $O(s_i)$ be the set of observations that could be the next observations of player $i$ in $\Gamma$.
    \begin{enumerate}
        \item For each player $i$ on team \pmax, let $\tilde o_i = o_i(h)$ be the true observation of player $i$. The adversary observes all recommendations $o_i(h)$ for players $i$ on its team. Then, for each such player, the adversary picks an observation $\tilde o_i \in O(s_i) \cup \{\bot\}$. Each $\tilde o_i$ is appended to the corresponding $s_i$. 
        \item The mediator observes $(\tilde o_1, \dots, \tilde o_n)$, and picks action recommendations $a_i \in A_i(\tilde o_i)$ to recommend to each player $i$, and appends each $a_i$ to the corresponding $s_i$.
        \item Players on team \pmax automatically play their recommended actions. The adversary observes all actions $(a_i : t_i = \pmin)$ recommended to players on team \pmin, and selects the action played by each member of team \pmin.
    \end{enumerate}
    The size of this game is given by the number of tuples of the form $(h, s_1, \dots, s_k)$ where $s_i$ is the sequence of adversary $i$ and $h$ is a node of the original game $\Gamma$. There are at most $|H|^{k+1}$ of these, so we are done.
\end{proof}

\section{Multi-Party Computation and Proof of Theorem~\ref{thm:main}}\label{sec:mpc}

We first formalize the usual setting of multi-party computation (MPC). Let $X$ be the set of binary strings of length $\ell$, and let $\lambda$ be a security parameter. \citet{Rabin89:Verifiable} claims in their Theorem 4 that essentially any protocol involving a mediator can be efficiently simulated without a mediator so long as more than half the players follow the protocol and we allow some exponentially small error. However, they do not include a proof of this result. In the interest of completeness, we prove the version of their result that is needed for our setting, based only on the primitives of {\em secure multi-party computation} and {\em verifiable secret sharing}.

\subsection{Secure MPC}
In {\em secure MPC}, there is a (possibly randomized) function $f : X^n \to X^n$ defined by a circuit with $N$ nodes. A subset $K \subset [n]$ of size $< n/2$ has been {\em corrupted}. Each {\em honest} player $i \in [n] \setminus K$ holds an input $x_i \in X$. The goal is to design a randomized messaging protocol, with {\em private} communication, such that, regardless of what the corrupted players do, there exist inputs $\{x_j : j \in K\}$ such that:
\begin{enumerate}
    \item (Output delivery) At the end of the protocol, each player $i$ learns its own output, $y_i := f(x_1, \dots, x_n)_i$, with probability $1 - 2^{-\lambda}$.
    \item (Privacy) No subset of $< n/2$ players can learn any information except their own output $y_i$. That is, the players cannot infer any extra information from analyzing the transcripts of the message protocol than what they already know.
    
    This can be intuitively modeled as follows. If any minority of colluded players were to analyze the empirical distributions of transcripts of the protocol for any input-output tuple, then such distribution would be fully explained in terms of their input-outputs, leaving no conditional dependence on other players' input-outputs.
    Formally, for any such subset $K \subseteq [n]$ of size $< n/2$, there exists a randomized algorithm ${\sf Sim}_K$ that takes the inputs and outputs of the players in set $K$, and reconstructs transcripts $T$, such that for all $x \in X^n$, we have
    \begin{align}
       \sum_T \abs{\Pr[{\sf Sim}_K(x_K, y_K)=T] - \Pr[{\sf View}_K(x)=T]} \le 2^{-\lambda}
    \end{align}
    where ${\sf View}_K(x)$ is the distribution of transcripts observed by players in set $K$ when running the protocol with input $x$.
\end{enumerate}

In other words, the players in set $K$ cannot do anything except pass to $f$ inputs of their choice. For our specific application, this implies that introducing an MPC protocol to simulate the mediator is equivalent to having a mediator, because no extra information aside from the intended function will be leaked to the players.
\begin{theorem}[\cite{Beaver90:Multiparty,Rabin89:Verifiable}]
    Secure MPC is possible, with polytime (in $\ell, \lambda, n, N$) algorithms that take at most polynomially many rounds and send at most polynomially many bits in each round.
\end{theorem}

\subsection{Verifiable Secret Sharing}

In the {\em verifiable secret sharing} (VSS) problem, the goal is to design a function ${\sf Share} : X \to {\bar X}^n$ , where ${\bar X} = \{0, 1\}^{\poly(\ell, \lambda, n)}$, that {\em shares} a secret $x \in X$ by privately informing each player $i$ of its piece ${\sf Share}(x)_i$, such that:
\begin{enumerate}
    \item{} (Reconstructibility) Any subset of $>n/2$ players can recover the secret fully, even if the remaining players are adversarial. That is, there exists a function ${\sf Reconstruct} : X^n \to X$ such that, where $(x_1, \dots, x_n) \gets {\sf Share}(x)$, we have ${\sf Reconstruct}(x_1', \dots, x_n') = x$ so long as $x_i' = x_i$ for $ > n/2$ players $i$.
        \item{} (Privacy) No subset of $< n/2$ players can learn any information about the secret $x$. That is, for any such subset $K$ and any secret $x$, there exists a distribution ${\sf Sim}_K \in \Delta(X^n)$ such that
        \begin{align}
            \norm{{\sf Share}(x)_K - {\sf Sim}_K}_1 \le 2^{-\lambda},
        \end{align}
        where $\norm{\cdot}_1$ denotes the $\ell_1$-norm on probability distributions. 
\end{enumerate}
\begin{theorem}[\cite{Rabin94:Robust,Rabin89:Verifiable}]
    There exist algorithms {\sf Share} and {\sf Reconstruct}, with runtime $\poly(\ell, \lambda, n)$, that implement robust secret sharing.
\end{theorem}
VSS is a primitive used in a fundamental way to build secure MPC protocols; to be formally precise in this paper, we will require VSS as a separate primitive as well to maintain the state of the mediator throughout the game.

\subsection{Simulating a Mediator}

We will assume that the game $\Gamma_0$ in \Cref{cor:zero-sum-efg} has been solved, and that its solution is given by a (possibly randomized) function
\begin{align}
    f : \Sigma \times O^n \to \Sigma \times A^n\label{eq:mediator}
\end{align}
where $\Sigma$ is the set of information states of the mediator in $\Gamma_0$. At each step, the mediator takes $n$ observations $o_1, \dots, o_n$ and its current infostate $s$ as input, and updates its infostate and outputs action recommendations according to the function $f$.

Consider the function $\hat f : (\bar X \times O)^n \to (\bar X\times A)^n$ defined by 
\begin{align}
    \hat f((x_1, o_1), \dots, (x_n, o_n)) = ((x_1', a_1), \dots, (x_n', a_n))
\end{align}
where
\begin{align}
    (s', a_1, \dots, a_n) = f({\sf Reconstruct}(x_1, \dots, x_n), o_1, \dots, o_n)
\end{align}
\begin{align}
(x'_1, \dots, x'_n) = {\sf Share}(s').
\end{align}
That is, $\hat f$ operates the mediator with its state secret-shared across the various players. The players will run secure MPC on $\hat f$ at every timestep. By the properties of MPC and secret sharing, this securely implements the mediator in such a way that the players on team \pmin can neither break privacy nor cause the protocol to fail, with probability better than $O(|H| \cdot 2^{-\lambda})$, where $H$ is the set of histories of the game.

We have thus shown our main theorem, which is more formally stated as follows:
\begin{theorem}[Formal version of \Cref{thm:main}]\label{thm:main-app}
    Let:
    \begin{itemize}
        \item $\Gamma$ be a hidden-role game with a \pmin-team of size $k < n/2$,
        \item $\Gamma^*$ be identical to $\Gamma$ except that there is an additional player who takes no nontrivial actions but is always on team \pmax;
        \item $\Gamma_0$ be the zero-sum game defined by \Cref{cor:zero-sum-efg} based on $\Gamma^*$;
        \item $x$ be any strategy of the \pmax player (mediator) in $\Gamma_0$, represented by an arithmetic circuit $f$ as in \eqref{eq:mediator} with $N = \poly(|H|^k)$ gates; and
        \item $\lambda$ be a security parameter.
    \end{itemize}
    Then there exists a strategy profile $x'$ of the \pmax players in $\gsplit(\gcommpriv^{M, R}(\Gamma))$, where $\log M = R = \poly(\lambda, N, \log|H|)$ such that:
    \begin{enumerate}
        \item  {\em (Equivalence of value)} the value of $x'$ is within $2^{-\lambda}$ of the value of $x$ in $\Gamma_0$, and
        \item  {\em (Efficient execution)} there is a $\poly(r)$-time randomized algorithm $\mc A_\Gamma$ that takes as input an infostate $s_i$ of $\gsplit(\gcommpriv^{M, R}(\Gamma))$ that ends with an observation, and returns the (possibly random) action that player $i$ should play at $s_i$.
    \end{enumerate}
\end{theorem}

\section{Connection to Communication Equilibria}\label{app:comm}
In this section, we prove \Cref{th:comm}, recalled below.

\thComm*

Before proving this result, we must first formally define a communication equilibrium. Given an arbitrary game $\Gamma$ with $n$ players, consider the $(n+1)$-player game in which the extra player is the mediator. Consider a private-communication extension $\tilde\Gamma$ of that game, with $R=2$ rounds and $M=|H|$, in which only communication with the mediator is allowed. In $\Gamma^*$, each player has a {\em direct strategy} $x_i^*$ in which, at every timestep, the player sends its honest information to the mediator, interprets the mediator's message in reply as an action recommendation, and plays that action recommendation. 

\begin{definition}
    A {\em communication equilibrium} of $\Gamma$ is a strategy profile $\mu=(x_1, \dots, x_n)$ for the mediator of $\tilde\Gamma$ such that, with $\mu$ held fixed, the profile  $(x_1, \dots, x_n)$ is a Nash equilibrium of the resulting $n$-player game.
\end{definition}
The {\em revelation principle for communication equilibria}~\cite{Forges86:Approach,Myerson86:Multistage} states that, without loss of generality in the above definition, it can be assumed that $x=x^*$, \ie, players are direct in equilibrium.

We now prove the theorem. Let $\Gamma^*$ be $\Gamma$ with a mediator who is always on team \pmax, and let $\Gamma_0$ be the zero-sum game constructed by \Cref{cor:zero-sum-efg}. Because $\hrvaluepriv(\Gamma)$ is the zero-sum value of $\Gamma_0$ by \Cref{thm:main}, it is enough to prove the inequality chain
\begin{align}
    \hrvaluepriv(\Gamma) \le {\sf CommVal}(\Gamma) \le {\sf Val}(\Gamma_0),
\end{align}
where {\sf CommVal} is the value of the \pmax-optimal communication equilibrium and {\sf Val} is the zero-sum game value.

By \Cref{thm:main}, the hidden-role equilibria of $\Gamma$ are (up to an arbitrarily small error $\eps > 0$) TMEs of $\gsplit(\Gamma^*)$. By \citet{vonStengel97:Team}, since \pmin has only one player, the TMEs of $\gsplit(\Gamma^*)$ are precisely the \pmax-team-optimal Nash equilibria of $\gsplit(\Gamma^*)$. Thus, for a TME $(\mu, x_1, \dots, x_n)$ of that game (where $\mu$ is a strategy of the mediator), there exists an adversary strategy $y$ such that $(\mu, x_1, \dots, x_n, y)$ is a Nash equilibrium. But then $(\mu, x_1, \dots, x_n, y)$ is also a communication equilibrium. Thus ${\sf TMEVal}(\gsplit(\Gamma^*)) \le {\sf CommVal}(\Gamma)$.

We now show the second inequality. If $\Gamma$ is an adversarial hidden-role game, each player $i$'s strategy can be expressed as a tuple $(x_i, y_i)$ where $x_i$ is player $i$'s strategy in the subtree where $i$ is on team \pmax, and $y_i$ is the same on team \pmin. In that case, the problem of finding a \pmax-optimal communication equilibrium can be expressed as:
\begin{align}
    \max_\mu \quad & u(\mu, x^*, y^*) \\
    \qq{s.t.} & \forall i~~\max_{x_i} u(\mu, x_i, x^*_{-i}, y^*) \le u(\mu, x^*, y^*) \\
     & \forall i~~\min_{y_i} u(\mu, x^*, y_i, y^*_{-i}) \ge u(\mu, x^*, y^*)
\end{align}
where $u$ is the \pmax-team utility function.
That is, no player can increase their team's utility by deviating from the direct profile $(x^*, y^*)$, regardless of which team they are assigned to. Now, using the fact that \pmin has only one player, we can write \pmax's utility function $u$ as a sum $u = \sum_i u_i$ where $u_i$ is the utility of \pmax when the \pmin-player is player $i$ (weighted by the probability of that happening). Each term $u_i$ depends only on $x$ and the $y_i$s. Thus, the above program can be rewritten as
\begin{align}
    \max_\mu \quad & \sum_i u_i(\mu, x^*, y^*_i) \\
    \qq{s.t.} & \forall i~~\max_{x_i} u(\mu, x_i, x^*_{-i}, y^*) \le u(\mu, x^*, y^*) \\
     & \forall i~~\min_{y_i} u_i(\mu, x^*, y_i) \ge u_i(\mu, x^*, y^*_i)
\end{align}
or, equivalently,
\begin{align}
    \max_\mu \quad & \sum_i \min_{y_i} u_i(\mu, x^*, y_i) \\
    \qq{s.t.} & \forall i~~\max_{x_i} u(\mu, x_i, x^*_{-i}, y^*) \le u(\mu, x^*, y^*).
\end{align}
or, equivalently,
\begin{align}
    \max_\mu \min_y \quad & u(\mu, x^*, y) \\
    \qq{s.t.} & \forall i~~\max_{x_i} u(\mu, x_i, x^*_{-i}, y^*) \le u(\mu, x^*, y^*).
\end{align}
This is precisely the problem of computing a zero-sum equilibrium in the game $\Gamma_0$, except with an extra constraint, so the inequality follows. \qed

\section{Complexity Bounds and Proofs}\label{sec:complexity-proofs}

Here we state and prove the various lower bounds in \Cref{tab:complexity-results}. Before proceeding, we make several remarks about the conventions used in this section.
\begin{itemize}
    \item Utilities will be given {\em unnormalized} by chance probability. That is, if we say that a player gets utility $1$, what we really mean is that the player gets utility $1/p$, where $p$ is the probability that chance sampled all actions on the path to $z$. Thus the contribution to the expected value from this terminal node will be $1$. This makes calculations easier. 
    \item The utility range of the games used in the reductions will usually be of the form $[-M, M]$ where $M$ is large but polynomial in the size of the game. Our definition of an extensive-form game allows only games with reward range $[-1, 1]$. This discrepancy is easily remedied by dividing all utility values in the proofs by $M$.
    \item $\Theta(\cdot)$ hides only an absolute constant.
\end{itemize}
To recap, we show hardness of {\em approximating} the value of the game to a desired precision $\eps > 0$. Our hardness results will hold even when $\eps = 1/\poly(|H|)$. 

\publichard*
\begin{proof}
    We show that given any graph $G$, it is possible to construct a hidden-role game based on $G$ whose value correspond to the size the graph's max-cut. This reduces MAX-CUT to the TME value problem.
    
    Let $G$ be an arbitrary graph with $n$ nodes and $m$ edges, and consider the following team game (no hidden roles). There are 3 players, 2 of whom are on team \pmax. The game progresses as follows.
    \begin{enumerate}
        \item Chance chooses two vertices $v_1, v_2$ in $G$, independently and uniformly at random. The two players on team \pmax observe $v_1$ and $v_2$ respectively.
        \item The two players on team \pmax select bits $b_1, b_2$, and the \pmin player selects a pair $(v_1', v_2')$. There are now several things that can happen: ($L$ is a large number to be picked later)
        \begin{enumerate}
            \item ({\em Agreement of players}) If $v_1 = v_2$ and $b_1 \ne b_2$, team \pmax gets utility $-L$.
            \item ({\em Objective}) If $v_1 \ne v_2$, $b_1 \ne b_2$, and $(v_1, v_2)$ is an edge in $G$, then team \pmax gets utility $1$.
            \item ({\em Non-leakage}) If $(v_1', v_2') = (v_1, v_2)$ then team \pmax gets utility $-(n^2-1)L$. Otherwise, team \pmax gets utility $L$.\footnote{Note that \pmin playing uniformly at random means that the expected utility of this term is $0$.}
        \end{enumerate}
    \end{enumerate}
    Consider a sufficiently large $L = \poly(m)$.
    The game is designed in such a way that \pmin's objective is to guess the vertices $v_1, v_2$ sampled by chance, but she has no information apart from the transcript of communication to guess it. Therefore, \pmin's optimal strategy is to punish any communication attempt between the players and play the most likely pair of vertices $v_1, v_2$. If no communication happens, her best strategy is to play a random pair of vertices.
    On the other hand, \pmax optimal strategy must ensure that under no circumstance players play the same bit when assigned to the same vertex (lest they incur the large penalty $L$). Therefore, the strategy of both \pmax players is to play a fixed bit in each vertex, and the optimal strategy is the one that assigns a bit to the vertices in such a way that the number of edges connecting vertices with different bits is maximized. This corresponds to finding a max cut and therefore the value of the game is (essentially) $c^*$ where $c^*$ is the true size of the maximum cut.
    Moreover, any communication attempt would be immediately shut down by \pmin strategy since any leak of the observation received on the public channel implies to receive a fraction of the large penalty $L$.
    
    We now formalize this intuition.
    
    First, note that \pmax can achieve utility exactly $c^*$ by playing according to a maximum cut. To see that \pmax cannot do significantly better, consider the following strategy for team \pmin. Observe the entire transcript $\tau$ of messages shared between the two \pmax players. Pick $(v_1, v_2)$ maximizing the probability $p(\tau | v_1, v_2)$ that the players would have produced $\tau$.\footnote{Note again, as in \Cref{sec:example}, that this computation may take time exponential in $r$ and the size of $G$, but we allow the players to perform unbounded computations. } 

    First, suppose that \pmax does not communicate. Then \pmin's choice is  independent of \pmax's, so \pmax can WLOG play a pure strategy. If $b_1 \ne b_2$ for any pair $(v_1, v_2)$, then \pmax loses utility $L$ in expectation. For $L > n^2$, this makes any such strategy certainly inferior to playing the maximum cut. Therefore, \pmax should play the maximum cut, achieving value $c^*$.

    Now suppose that \pmax uses communication. Fix a transcript $\tau$, and let
    \begin{align}
        \delta := \max_{v_1, v_2} p(v_1, v_2|\tau) - \frac{1}{n^2}.
    \end{align}
    Now note that we can write $p(\cdot|\tau) = (1 - \alpha) q_0 + \alpha q_1$, where $q_0, q_1 \in \Delta(V^2)$, $q_0$ is uniform, and $\alpha \le \Theta(n^4\delta)$. Now consider any strategy that \pmax could play, given transcript $\tau$. Such a strategy has the form $x_1(b_1 | v_1)$ and $x_2(b_2|v_2)$. Since $q$ is $\alpha$-close to uniform, the utility of \pmax under this strategy conditioned on $\tau$ must be bounded above by
    $$ (1 - \alpha)u_0 + \alpha \le (1 - \alpha) c^* + \alpha \le c^* + \alpha \le c^* + \Theta(n^4 \delta)$$
    where $u_0$ is the expected value of profile $(x_1, x_2)$ given $\tau$ if $v_1, v_2|\tau$ were truly uniform.
    But now, in expectation over $\tau$, team \pmin can gain utility $Ln^2 \delta$ by playing $\argmax_{v_1, v_2} p(v_1, v_2|\tau)$. So, \pmax's utility is bounded above by 
    $$ c^* + \Theta(n^4 \delta) - Ln^2 \delta \le c^*$$ by taking $L$ sufficiently large.
\end{proof}

The next result illustrates the difference between the uncoordinated hidden-role value and the coordinated hidden-role value which is the focus of the positive results in our paper. Whereas the coordinated hidden-role value with private communication can be computed in polynomial time when $k$ is constant (\Cref{thm:main}), the uncoordinated hidden-role value cannot, even when $k=2$:

\begin{theorem}\label{thm:private hard}
    Even in $3$-vs-$2$ hidden-role games, the uncoordinated hidden-role value problem with private communication is \coNP-hard.
\end{theorem}
\begin{proof}
    We reduce from UNSAT. Let $\phi$ be any 3-CNF-SAT formula, and consider the following $5$-player hidden-role game. Two players are chosen uniformly at random to be on team \pmin; the rest are on team \pmax. The players on team \pmin know each other. The players on team \pmin play the SAT gadget game described by \citet{Koller92:Complexity}. Namely:
    \begin{enumerate}
        \item The players on team \pmin are numbered P1 and P2, at random, by chance.
        \item Chance selects a clause $C$ in $\phi$ and tells P1.
        \item P1 selects a variable $x_i$ in $C$, and that variable (but not its sign in $C$, nor the clause $C$ itself) is revealed to P2.
        \item P2 selects an assignment $b_i \in \{0, 1\}$ to $x_i$. \pmin wins the gadget game if the assignment $b_i$ matches the sign of $x_i$ in $C$.
    \end{enumerate}
    The value of this game is decreasing with $M$ and $R$ since \pmax does nothing, so it is in the best interest of \pmax to select $R = 0$ (\ie, allow no communication). In that case, the best probability with which \pmin can win the game is exactly the maximum fraction of clauses satisfied by any assignment, which completes the proof.
\end{proof}

The above result is fairly straightforward: it is known that optimizing the joint strategy of a team with asymmetric information\footnote{For our purposes, we will say that \pmin has {\em symmetric information} if all players \pmin have the same observation at every timestep. This implies that they can be merged into a single player without loss of generality.} is hard~\cite{Koller92:Complexity}, and private communication does not help if \pmax does not allow its use. However, next, we will show that the result even continues to apply when \pmin has {\em symmetric} information, that is, when the original game $\Gamma$ is coordinated. This may seem mysterious at first, but the intuition is the following. Just because $\Gamma$ has symmetric information for the \pmin-team, does not mean $\gsymsplit(\gcommpriv(\Gamma))$ does. Indeed, \pmax-players can send different private messages to different \pmin-players, resulting in asymmetric information among \pmin-players. This result illustrates precisely the reason that we define two different split-personality games, rather than simply dealing with the special case where the original game $\Gamma$ has symmetric information for the \pmin-team.

\begin{theorem}\label{thm:uninformed hard}
    Even in $3$-vs-$2$ hidden-role games with a mediator in which no \pmin-player has any information beyond the team assignment, the uncoordinated hidden-role value problem with private communication is \coNP-hard.
\end{theorem}

\begin{proof}
    We will reduce from (the negation of) MAX-CUT. Let $G$ be an arbitrary graph with $n$ nodes and $m$ edges, and consider the following $5$-player hidden-role game with $2$ players on the \pmin-team and $3$ players on the \pmax-team. Player 5 is always on team \pmax and is the mediator. The other four players are randomly assigned teams so that two are on team \pmax and two are on team \pmin. The game proceeds as follows.
    \begin{enumerate}
        \item Chance chooses vertices $v_1, \dots, v_4$ uniformly at random from $G$. The mediator privately observes the whole tuple $(v_1, \dots, v_4)$.
        \item For notational purposes, call the players on team \pmax 3 and 4, and \pmin 1 and 2. (The mediator does not know these numbers.) After some communication, the following actions happen simultaneously:
        \begin{enumerate}
            \item P1 and P2 select bits $b_1, b_2 \in \{0, 1\}$;
            \item P3 and P4 select vertices $v_3', v_4'$ of $G$ and players $i_3, i_4 \in \{1, 2, 3, 4\}$.
        \end{enumerate}
        \item The following utilities are given: ($L$ is a large number to be picked later, and each item in the list represents an additive term in the utility function)
        \begin{enumerate}
            \item ({\em Correct vertex identification}) For each player $i \in \{3, 4\}$, if $v_i \ne v_i'$ then \pmax gets utility $-L^4$
            \item ({\em Agreement of \pmin-players}) If $v_1 = v_2$ and $b_1 \ne b_2$ then team \pmax gets utility $L$.
            \item ({\em Objective}) If $v_1 \ne v_2$, $b_1 \ne b_2$, and $(v_1, v_2)$ is an edge in $G$, then team \pmax gets utility $-1$.
            \item ({\em Privacy}) For each $j \in \{3, 4\}$, if $i^*_j$ is on team \pmin, then \pmax gets utility $L$. Otherwise, \pmax gets utility $-L/2$.
        \end{enumerate}
    \end{enumerate}

    We claim that this game has value (essentially) $-c^*$, where $c^*$ is the actual size of the maximum cut of $G$. To see this, observe first that the mediator {\em must} tell all players their true vertices $v_i$, lest it risk incurring the large negative utility $-L^2$. Further, any player except the mediator who sends a message must be on team \pmin. This prevents team \pmin from communicating. Thus, the mediator's messages force \pmin to play an asymmetric-information identical-interest game, which is hard.

    We now formalize this intuition. First, consider the following strategy for team \pmax: The mediator sends all players their true types, and \pmax-players play their types. If any \pmax-player sees a message sent from anyone except the mediator, the \pmax-player guesses that that player is on team \pmin. 
    
    Now consider any (pure) strategy profile of team \pmin. First, \pmin achieves utility $-c^*$ by observing the mediator's message and playing bits according to a maximum cut.  We now show that this is the best that \pmin can do. Sending messages is, as before, a bad idea. Thus, a pure strategy profile of \pmin is given by four  $f_1, f_2, f_3, f_4 : V \to \{0, 1\}$ denoting how player $i$ should pick its bits. But then $f_1 = f_2 = f_3 = f_4$; otherwise, the agreement of \pmin-players would guarantee that \pmin is not playing optimally for large enough $L$.

    Now, for any (possibly mixed) strategy profile of team \pmax, consider the following strategy profile for each \pmin-player. Let $f : V \to \{0, 1\}$ be a maximum cut. Pretend to be a \pmax-player, and let $v_i'$ be the vertex that would be played by that \pmax-player. Play $f(v_i')$. 

    First, consider any \pmax-player strategy profile for which, for some player $i$ and some $v_i$, the probability that $v_i' \ne v_i$ exceeds $1/L^2$. Then \pmax gets a penalty of roughly $L^2$ in expectation, but now setting $L$ large enough would forces \pmax to have utility worse than $-1$, so that \pmax would rather simply play $v_i$ with probability $1$.

    Now, condition on the event that $v_i = v_i'$ for all $i$ (probability at least $1 - \Theta(1/L^2)$). In that case, the utility of \pmax is exactly $-c^*$, because \pmin is playing according to the maximum cut. Thus, the utility of \pmax is bounded by $-(1 - \Theta(1/L^2)) c^* + \Theta(1/L^2 \cdot L) \le -c^*/n^2 + \Theta(n/L) < c^* + 1/2$ for sufficiently large $L$. Thus, solving the hidden-role game to sufficient precision and rounding the result would give the maximum cut, completing the proof.
\end{proof}

\begin{theorem}
    Even in $5$-vs-$4$ hidden-role games, the uncoordinated hidden-role value problem with public communication is $\Sigma_2^\P$-hard. 
\end{theorem}
\begin{proof}
    We reduce from $\exists\forall$3-DNF-SAT, which is $\Sigma_2^\P$-complete~\cite{Haastad01:Some}. The $\exists\forall$3-DNF-SAT problem is the following. Given a 3-DNF formula $\phi(x, y)$ with $k$ clauses, where $x \in \{0, 1\}^m$ and $y \in \{0, 1\}^n$, decide whether $\exists x \forall y\ \phi(x, y)$. Consider the following game. There are 9 players, 5 on team \pmax and 4 on team \pmin. One designated player, who we will call P0, is \pmax and has no role in the game. (The sole purpose of this player is so that \pmax is a majority.) The other players are randomly assigned teams. These other players are randomly assigned teams. For the sake of analysis, we number the remaining players P1 through P8 such that P1, P3, P4, P5 are on team \pmax and P2, P6, P7, P8 are on team \pmin. \pmin knows the entire team assignment, whereas \pmax dos not. We will call P3--P8 ``regular players'', and P1--P2 ``guessers''. The game proceeds as follows.
    \begin{enumerate}
        \item For each regular \pmax-player, chance selects a literal (either $x_j$ or $\neg x_j$), uniformly at random. For each regular \pmin-player (P6--8), chance selects a literal (either $y_j$ or $\neg y_j$), also uniformly at random. Each player privately observes the {\em variable} (index $j$), but not the sign of that variable.
        \item After some communication, the following actions happen simultaneously.
        \begin{enumerate}
            \item P3--P8 select assignments (0 or 1) to their assigned variables.
            \item P1 (who observes nothing) guesses a player (among the six players P3--P8) that P1 believes is on team \pmin.
            \item P2 (who observes nothing) guesses one literal for each \pmax-player (there are $K := (2m)^3$ such possible guesses.)
        \end{enumerate}
            \item The following utilities are assigned. ($L$ is a large number to be picked later, and each item in the list represents an additive term in the utility function)
        \begin{enumerate}
            \item ({\em Satisfiability}) Chance selects three regular players at random. If the three literals given to those players form a clause in $\phi$, and that clause is satisfied, \pmax gets utility 1.
            \item ({\em Consistency}) If chance selected the same variable three times, if the three players did not give the same assignment, then the team (\pmax if the variable was an $x_i$ and \pmin if the variable was a $y_j$) gets utility $-L^2$.
            \item ({\em Privacy for \pmax}) If P2 guesses the three literals correctly, \pmin gets utility $L^3K$; otherwise, \pmin gets utility $-L^3$.\footnote{These utilities are once again selected so that a uniformly random guess gets utility $0$.}
            \item ({\em Privacy for \pmin}) If P1 guesses a \pmin-player, \pmax gets utility $L$; otherwise, \pmax gets utility $-L$.
        \end{enumerate}
    \end{enumerate}
    We claim that the value of this game with public communication is at least $1$ if and only if $\phi$ is $\exists\forall$-satisfiable. Intuitively, the rest of the proof goes as follows:
    \begin{enumerate}
        \item \pmax will not use the public communication channels if $\phi$ is $\exists\forall$-satisfiable. By the {\em Privacy for \pmin} term, \pmin-players therefore cannot do so either without revealing themselves immediately. Thus, \pmax will get utility at least $1$ if $\phi$ is $\exists\forall$-satisfiable.
        \item If $\phi$ is not $\exists\forall$-satisfiable, then consider any \pmax-team strategy profile. By the {\em Privacy for \pmax} term, team \pmax cannot make nontrivial use of the public communication channel without leaking information to P9. By the {\em Consistency} term, P1 through P3 must play the same assignment, or else incur a large penalty. So, \pmax must play essentially an assignment to the variables in $x$, but such an assignment cannot achieve positive utility because there will exist a $y$ that makes $\phi$ unsatisfied. 
    \end{enumerate}
    We now formalize this intuition. Suppose first that $\phi$ is $\exists\forall$-satisfiable, and let $x$ be the satisfying assignment. Suppose that \pmax-players never communicate and assign according to $x$, and P4 guesses any player that sends a message. Then any \pmin-strategy that sends a message is bad for sufficiently large $L$ because it guarantees a correct guess from \pmax; any \pmin-strategy that is inconsistent is bad because it will lose utility at least $L$; and any \pmin-strategy that is consistent will cause \pmax to satisfy at least one clause. Thus \pmax guarantees utility at least $1$.

    Now suppose that $\phi$ is not $\exists\forall$-satisfiable. Consider any strategy profile for \pmax. Suppose that the \pmin-players play as follows. During the public communication phase, each \pmin-player samples a literal from the set $\{ x_1, \neg x_1, \dots, x_m, \neg x_m \}$ uniformly at random and pretends to be a \pmax-player given that literal. P8 observes the public transcript, and selects the triplet of literals that is conditionally most likely given the transcript. By an identical argument to that used in \Cref{thm:public-hard}, \pmax then cannot profit from using the communication channel for sufficiently large $L$. Therefore we can assume that \pmax does not use the communication channels, and therefore by the argument in the previous paragraph, neither does \pmin. 

    Now, the strategy of each \pmax-player $i$ can be described by a vector $s_i \in [0, 1]^m$, where $s_{ij}$ is the probability that player $i$ assigns $1$ to variable $x_t$. For sufficiently large $L$, we have $s_i \in [0, \eps] \cup [1-\eps, 1]$ where $\eps = 1/L$, because otherwise \pmax would incur a penalty proportional to $L^2 \eps > k$ and would rather just play (for example) the all-zeros profile, which guarantees value $0$. Condition on the event that every player at every variable chooses to play the most-likely assignment according to the $s_i$s. This happens with probability at least $1 - \Theta(m \eps)$. These assignments must be consistent (\ie, every player must have the same most-likely assignment), or else the players would once again incur a large penalty proportional to $L^2$. Call that assignment $x$, and let $y$ be such that $\phi(x, y)$ is unsatisfied. Suppose \pmin plays according to $y$. Then \pmax's expected utility is bounded above by $\Theta(m \eps k)$: with probability $1-\Theta(m\eps)$ it is bounded above by $0$; otherwise it is bounded above by $k$. For $\eps < \Theta(1/mk)$ this completes the proof.
\end{proof}

\section{The Game {\em Avalon}}\label{sec:avalon}

\subsection{Equivalence of Split-Personality Games}\label{sec:avalon:split-personality}

In this section, we show that the choice of whether to use \gsymsplit or \gsplit---that is, whether \pmin is coordinated---is irrelevant for the instantiations of {\em Avalon} that we investigate in this paper. We refer to \Cref{sec:experiments} for a complete description of the \emph{Avalon} game. 

For this section, we assume that all \pmin players know the roles of all other \pmin players. This is always true in the instances considered in the experiments. 
In particular, we considered instances with six or fewer players and no \emph{Oberon} role, which is a \pmin member which is not revealed as such to both \emph{Merlin} and the other \pmin members. This guarantees that both the \pmin-players can deduce their respective roles. Conversely this is not guaranteed, as for example with 7 players (and hence $k=3$) and Mordred, the two non-Mordred \pmin-players do not know the identity of Mordred.

The main observation we make in this subsection states that the choice of whether to use  \gsymsplit or \gsplit does not matter.

\begin{theorem}\label{th:avalon}
    In {\em Avalon} with private communication, assuming that all \pmin players know each others' roles, every $r$-round uncoordinated hidden-role equilibrium is also an $r$-round coordinated hidden-role equilibrium.
\end{theorem}

Before proving the result, we first observe that we cannot {\em a priori} assume the use the results in \Cref{sec:one-sided} for this proof, because they do not apply to the symmetric hidden-role equilibrium. 

Fortunately, a variant of the revelation principle for mediated games (\Cref{th:revelation principle}) and the resulting zero-sum game theorem (\Cref{cor:zero-sum-efg}) still holds {\em almost} verbatim: the lone change is that the zero-sum game formed from symmetric splitting, which we will denote $\tilde\Gamma_0$, is a zero-sum {\em team} game, where the \pmin players cannot perfectly coordinate with each other, and the equality of value becomes an inequality:
\begin{proposition}\label{prop:rp-zs-uncoord}
    Let $\Gamma^*$ be a mediated hidden-role game, and let $\Gamma_0$ be the zero-sum version posed by \Cref{cor:zero-sum-efg}. Let $\tilde\Gamma_0$ be the game that is identical to $\Gamma_0$, except that the \pmin-players are not coordinated and thus cannot communicate except as permitted in $\Gamma$. That is, $\tilde\Gamma_0$ is a {\em team game} with the mediator as the \pmax-player and the adversaries as the \pmin-players. Then $\symhrvalue_{\sf priv}(\Gamma^*) \le {\sf TMEVal}(\tilde \Gamma_0)$, where {\sf TMEVal} is the TME value function (with \pmax committing first).\footnote{Since \pmin is the nontrivial team and \pmax commits first, the TME and TMECor values of $\tilde \Gamma_0$ coincide here.}
\end{proposition}
\begin{proof}
    The revelation principle (\Cref{th:revelation principle}) applies verbatim in this setting, so we may assume that \pmax in $\Gamma^*$ uses the mediator as specified in the revelation principle. With this assumption, the only difference that remains between $\Gamma^*$ and $\tilde\Gamma_0$ is that, in the latter, \pmin-players cannot coordinate {\em at all}, whereas in $\Gamma^*$, \pmin-players are able to somewhat coordinate by sending private messages (albeit at the cost of revealing themselves as adversaries). But this is a strict disadvantage for \pmin in $\tilde\Gamma_0$.
\end{proof}

We therefore have the inequality chain
\begin{align}
   {\sf Val}(\Gamma_0) = \hrvalue_{\sf priv}(\Gamma) \le \symhrvalue_{\sf priv}(\Gamma) \le \symhrvalue_{\sf priv}(\Gamma^*) \le {\sf TMEVal}(\tilde \Gamma_0)
\end{align}
where $\Gamma_0$ is the zero-sum game that appears in \Cref{cor:zero-sum-efg}, with asymmetric splitting, {\sf Val} is the zero-sum game value, and {\sf Val} is the zero-sum value.

So far, this inequality chain would hold for {\em any} hidden-role game $\Gamma$. We will now show that, specifically for {\em Avalon}, we have ${\sf TMEVal}(\tilde \Gamma_0) = {\sf Val}(\Gamma_0)$, which would complete the proof of \Cref{th:avalon}. This part of the proof depends on the special structure of {\em Avalon}.

The only difference between $\Gamma_0$ and $\tilde\Gamma_0$ is that, in the latter, \pmin-players cannot communicate among each other because \Cref{cor:zero-sum-efg} only consider communication with the mediator. But, what would \pmin-players even need to communicate? Their information is entirely symmetric, except that they do not know the recommendations that were given to their teammates in the mission proposal and voting phases. We will show the following result.
\begin{lemma}\label{lem:unilateral proposals}
    For both $\tilde \Gamma_0$ and $\Gamma_0$, the game value does not change if we make the assumption that the mediator unilaterally dictates missions (circumventing the mission proposal and voting process).
\end{lemma}
\begin{proof}[Proof Sketch]
Make the following restrictions to strategy spaces.
     \begin{itemize}
         \item The mediator picks a single mission proposal in each round, and recommends that proposal and that everyone vote in favor of it, until it is approved, and ignores any information gained during this phase.
         \item {\em Spies} follow the above recommendations, and ignore all information except the eventually-approved mission.
     \end{itemize}
     Consider any TMECor of the restricted game. We claim that it also must be a TMECor of the full game:
     \begin{itemize}
         \item The mediator cannot do better, because {\em Spies} are following all recommendations and ignoring any information gained except the eventual actual mission.
         \item {\em Spies} cannot do better, because deviations are ignored by the mediator and will ultimately not affect the approved mission. (because the number of mission proposals and the number of {\em Resistance} players both exceed the number of {\em Spies}). 
     \end{itemize}
     This completes the proof.
\end{proof}
But, having restricted the strategy spaces in the above manner, $\tilde\Gamma_0$ and $\Gamma_0$ become identical, because there is no other source of asymmetric information. Therefore, ${\sf TMEVal}(\tilde \Gamma_0) = {\sf Val}(\Gamma_0)$, and we are done.

\subsection{Abstractions}\label{sec:avalon:abstractions}

In this section we describe the abstractions employed on the original game in order to transforming into a smaller version, which can then be solved tabularly. The abstractions employed are lossless, in the sense that they do not change the value of the game, and that any Nash equilibrium in the abstracted game corresponds to a Nash equilibrium in the original one.

In particular, our main interest in the experiments is to compute one strategy belonging to a Nash equilibria for each player. This allows us to iteratively remove many actions without loss of generality.

Note that throughout this section, we interchangeably use the standard names {\em Resistance} and {\em Spies} to refer to the teams instead of \pmax and \pmin.

We apply \Cref{th:revelation principle,cor:zero-sum-efg,thm:main} to an Avalon game instance in order to compute a hidden-communication equilibria. In the first communication round, anybody can claim to the mediator to be a specific role and to know specific knowledge about other player's roles. (For example, a \pmin-player, that is, a {\em Spy}, could claim to be Merlin and lie about who is on team \pmax.) Thanks to the revelation principle (\Cref{th:revelation principle}), we can assume that each player in the \emph{Resistance} team truthfully communicates their information, while each \emph{Spies} player may decide correlate with the others to emit a specific fake claim. During the game, the mediator recommends an action to each player anytime it is their move. In this case, the revelation principle allows to assume \wlo that \emph{Resistance} players always obey the recommendations received from the mediator, while \emph{Spies} can decide to deviate. Overall we obtain a two-player zero-sum game in which the \emph{Resistance} team is represented by a mediator and the \emph{Spies} team is a single player.

In the following, we list the abstractions we applied to our Avalon instances, and for each we sketch a proof of correctness. Each proof will have the same structure: we will impose a restriction on the strategy spaces of both teams, resulting in a smaller zero-sum game; we will then show that any equilibrium of this smaller zero-sum game cannot have a profitable deviation in the full game.

 \begin{enumerate}
     \item {\em The recommendations emitted by the mediator to the player in charge of the mission proposal should never be rejected. That is, the mediator has unilateral power to dictate who goes on missions.}

     This is \Cref{lem:unilateral proposals}.

     \item {\em Call the following information {\em common knowledge}:
     \begin{itemize}
         \item The sets of good players claimed by players claiming to be Merlin, but {\em not} the identities of the claimants themselves (because \pmin does not know the true Merlin)
        \item Mission proposals and results
         \item Any claimant whose claim is disproven by common-knowledge information is a {\em Spy}
     \end{itemize}
     
     Call a subset $S \subseteq [n]$ {\em plausible} if, based on common-knowledge information, it is possible that $S$ contains only good players. Every mission proposal should be plausible.
     }

    \textit{Proof sketch.} Make the following restrictions to strategy spaces.
    \begin{itemize}
        \item Every mission proposal by the mediator is plausible.
        \item If the mediator proposes an implausible set, the {\em Spies} pick a plausible set at random, announce it publicly, and act as if the mediator proposed that set instead.
     \end{itemize}
    Consider any Nash equilibrium of the restricted game. We claim that it also must be a Nash equilibrium of the full game. 
    \begin{itemize}
        \item The mediator cannot improve, because {\em Spies} will pretend that the mediator played a plausible set anyway, and the mediator can simulate how the spies will do this.
        \item {\em Spies} cannot improve, because against a mediator who always proposes a plausible set, the strategy space of the spies is not limited.
    \end{itemize}

     \item {\em If any player is contained in {\em every} plausible set $S$, then that player should always be included on missions if possible. We call such players {\em safe}.}

     {\em Proof sketch.} Make the following restrictions  to strategy spaces.
     \begin{itemize}
                 \item The mediator always sends all {\em safe} players on missions, if possible.
                 \item {\em Spies} pretend that all {\em safe} players are always sent on missions. More precisely, if $i$ is safe, not sent on a mission, and at least one unsafe player $j$ is sent on that mission, {\em Spies} pick such a pair $(i, j)$ at random, announce both $i$ and $j$ publicly, and pretend that $i$ was sent in place of $j$. (If there are multiple such replacements that can be done, this process is simply repeated.)
     \end{itemize}
     Consider any Nash equilibrium of the restricted game. We claim that it also must be a Nash equilibrium of the full game. 
     \begin{itemize}
         \item The mediator cannot improve, because {\em spies} will pretend that all safe players are always sent regardless of how the mediator actually proposes missions, and the mediator can simulate how {\em spies} will perform the replacement.
         \item {\em Spies} cannot improve, because against a mediator who always sends all safe players on missions, the strategy space of the spies is not limited.
     \end{itemize}

     \item {\em If a mission of size $s$ passes, all missions will pass until the next mission whose size is $> s$.} (In particular, any maximum-size mission passing implies that all remaining missions pass.)

     {\em Proof sketch.} Suppose that a mission $M_0$ of size $s_0$ has passed, and suppose the next $t$ missions have sizes $s_1, \dots, s_t \le s_0$. Make the following restrictions to strategy spaces.
     \begin{itemize}
                 \item The mediator randomly\footnote{The fact that this is {\em random} instead of dictated by the mediator allows us to avoid issues of imperfect recall.} selects subsets $M'_i \subseteq M_0$ to send on missions $s_i$ for $i = 1, \dots, t$. If any of them (say, $M_i$) fails, the mediator pretends that $M_0$ failed instead. Then, the mediator generates and publicly announces the missions $M_1', \dots, M_i'$ that it {\em would have} proposed had $M_0$ failed, assuming that each $M_j'$ passes. The mediator then pretends that that is what happened. The mediator does not use the $M_i$s to inform future decisions.
                 \item For each $i = 1, \dots, t$, {\em Spies} automatically pass $M_i$ and ignore what missions are proposed by the mediator.
     \end{itemize}
     Consider any Nash equilibrium of the restricted game. We claim that it also must be a Nash equilibrium of the full game. 
    \begin{itemize}
        \item The mediator has no incentive to propose a mission different from a random $M_i \subseteq M_0$ in the unrestricted game, because against spies who ignore and always pass missions $M_1, \dots, t$, any mediator mission would lead to the same outcome of having a mission passing, no change in behavior from {\em Spies}, and no information gained.
        \item The \emph{Spies} have no incentive to fail mission $M_i$ for $i > 0$ if they pass mission $M_0$. This is true because the mediator strategy is such that \emph{Spies} passing missions $M_0, \dots, M_{i-1}$ and failing mission $M_i$ would be equivalent to \emph{Spies} failing $M_0$ and passing missions $M'_1, \dots M'_i$.
    \end{itemize}
 \end{enumerate}

     \citet{Christiano18:Solving} used a very strong reduction that allowed the assumption that the smallest two missions are guaranteed to pass. The reduction is sound in the base game {\em Resistance}, as well as when the mediator's strategy is manually constrained to {\em never} reveal information about Merlin, as is the case in Christiano's analysis. It turns out, however, that this reduction is {\em not} sound in general, because its proof assumes that all mission proposals---even those that occur after {\em Resistance} has already reached the required number of missions---are public information. However, in our setting, it is possible that each mission proposal leaks more and more information about Merlin, so there is incentive to reveal as few missions as possible. Indeed, in the 6-player variants with Mordred, we observe that removing the two smallest missions  causes a small but nonzero change in the game value, whether or not future mission proposals are published. However, a more refined analysis can be used to partially recover a reduction in the number of missions. \citet{Christiano19:Simplifying} described the correct reduction in the 6-player case, but without proof. Here, we formalize the argument.

In 5-player {\em Avalon}, the mission sizes are 2, 3, 2, 3, 3 in that order.
\begin{proposition}
    In 5-player Avalon with private communication, the hidden-role value is the same if the two missions of size 2 are automatically passed (hence skipped).
\end{proposition}
\begin{proof}[Proof Sketch]
    Make the following restrictions to strategy spaces.
    \begin{itemize}
        \item The mediator privately picks three missions $M_1, M_2, M_3$ of size $3$. Then the mediator proposes $M_1$ (or a random subset thereof) until it fails, then $M_2$ or a subset thereof until that fails, and finally $M_3$ until that fails.
        \item {\em Spies} ignore and automatically pass the two smallest mission proposals, and fail the three largest missions if possible.
    \end{itemize}
    Consider any Nash equilibrium of the restricted game. We claim that it also must be a Nash equilibrium of the full game.
    \begin{itemize}
        \item The mediator has no incentive to behave differently, because {\em Spies} are ignoring the two smallest missions regardless.
        \item {\em Spies} win if either three missions are failed, or they guess Merlin. But if at least one mediator guess is the true set of good players, only two missions can fail, and so {\em Spies} cannot increase their probability of failing three missions. As for guessing Merlin, {\em Spies} cannot increase the amount of information gained. Let $i$ be the index such that $M_i$ is the true good set, and $i=4$ if all mediator guesses were wrong. Spies are currently gaining the information $M_{\le i}$, and Spies cannot gain more information from deviating because the mediator always plays the missions in order.
    \end{itemize}
    This completes the proof.
\end{proof}

In six-player Avalon, the mission sizes are 2, 3, 4, 3, 4. The problem in adapting the analysis of the previous result is that the final size-3 mission comes {\em between} the size-4 missions, preventing its removal. 
\begin{proposition}
    In 6-player Avalon  with private communication, the hidden-role value is the same if the mission of size 2 is automatically passed (hence skipped).
\end{proposition}
\begin{proof}[Proof Sketch]
    Consider any strategy $x$ for the mediator in the restricted game in which the first mission is ignored. We will extend $x$ to a full-game strategy $x'$. Strategy $x'$ works as follows. The mediator maintains a simulator of strategy $x$.
    \begin{enumerate}
        \item The mediator queries the $x$-simulator for its second mission $M_2$ (of size 3), and plays a random size-2 subset of $M_2$. If this mission passes, the mediator plays $M_2$ itself on the second mission, then plays the rest of the game according to $x$.
        \item If the first mission fails, the mediator tells the $x$-simulator that $M_2$ failed. The $x$-simulator will then output its third mission $M_3$ (of size 3). The mediator picks a random size-3 subset of $M_3$ to send on the second mission. If this mission passes, the mediator plays $M_3$ again on the third mission, and once again plays the rest of the game according to $x$.
        \item If the second mission fails, the mediator tells the $x$-simulator that $M_3$ failed. The $x$-simulator will then output its fourth mission $M_4$ (of size 3). The mediator immediately tells the $x$-simulator that $M_4$ passed. The $x$-simulator will then output its fifth mission $M_5$, of size $4$. The mediator plays random subsets of $M_5$ until the game ends.
    \end{enumerate}
    We claim that, against this mediator, the adversary cannot do better by failing the first mission than by passing it. 

    \begin{itemize}
        \item If the adversary fails the first mission and passes the second, then the adversary will have, in the first three missions, observed a random subset of $M_2$ and the whole set $M_3$. But the adversary could have accomplished more by passing the first mission and failing $M_2$---the behavior of the mediator would be the same in both cases (by construction of the mediator), and the adversary would observe the whole set $M_2$ instead of merely a random subset.
        \item If the adversary fails the first two missions, the adversary (regardless of what else happens in the game) will have observed a random subset of $M_2$, a random subset of $M_3$, and $M_5$ by the end of the game. Further, three missions will pass if and only if $M_5$ is the true good set. The adversary, however, could have accomplished more by passing the first mission, failing the second and third ($M_2$ and $M_3$), and passing the fourth---the adversary induces the same outcome in the game, but instead observes all four missions $M_2, M_3, M_4, M_5$. 
    \end{itemize}
    Thus, the adversary is always better off failing the first mission, and therefore $x$ and $x'$ have the same value. 
\end{proof}

\subsection{A Description of {\em Avalon} in Reduced Representation}
The optimized \emph{Avalon} game instances can be succinctly represented as follows:
\begin{enumerate}
    \item At the start of the game, Chance player deals the roles.
    \item \textit{[if Merlin is present in the game]} Merlin reports who the {\em Spies} are (except {\em Mordred}). Spies may also (falsely) do the same.
    \item \textit{[if Percival is present in the game]} Percival reports who Merlin is. {\em Spies} may also falsely do the same.
    \item Until three missions are passed, the mediator sends a mission, and the \emph{Spies} are offered the option to fail it if any of them is on it. The space of possible missions to be proposed is constrained according to the previous optimizations.
    \item If three missions are failed, the game ends with a win for the \emph{Spies}.
    \item If three missions are passed and Merlin is present, the \emph{Spies} have to guess the player with the Merlin role. If they guess correctly, they win, otherwise they lose. If Merlin is absent, the game ends with a win for the \emph{Resistance}.
\end{enumerate}

\begin{table*}[t]
    \centering\scalebox{0.8}{
    \begin{tabular}{l|rrr|rrr}
       & \multicolumn{3}{c|}{\bf 5 Players} & \multicolumn{3}{c}{\bf 6 Players} \\
         \bf Variant & \multicolumn{1}{c}{$|Z|$} & \multicolumn{1}{c}{RW} & \multicolumn{1}{c|}{MG} & \multicolumn{1}{c}{$|Z|$} & \multicolumn{1}{c}{RW} & \multicolumn{1}{c}{MG} \\\toprule
         No special roles ({\em Resistance}) & $5.9 \times 10^3$ & $0.3000$ & {n/a} & $1.4 \times 10^6$ & $0.3333$ & {n/a} \\
         Merlin only & $1.5\times 10^4$ & $1.0000$ & $0.3333$ & $6.0 \times 10^5$ & $1.0000$ & $0.2500$ \\
         Merlin + Mordred & $4.9 \times 10^5$ & $0.6691$ & $0.3869$ & $1.8 \times 10^8$ & $0.7529$ & $0.3046$\\
         Merlin + 2 Mordreds & $9.0 \times 10^4$ & $0.5000$ & $0.4444$ & $2.8 \times 10^7$ & $0.4088$ & $0.2886$ \\
         Merlin+Mordred+Percival+Morgana & $2.4 \times 10^6$ & $0.9046$ & $0.3829$ & unknown & unknown & unknown  \\\bottomrule
    \end{tabular}
    }
    \caption{Breakdown of equilibrium outcome probabilities for each variant of {\em Avalon}. $|Z|$: number of terminal nodes in the reduced game tree. `RW': probability of Resistance passing three missions. `MG': probability of  Spies guessing Merlin correctly, conditioned on Resistance passing three missions. (The game value, which appears in \Cref{tab:experiments}, is $\text{RW} \cdot (1 - \text{MG})$). 
    }
    \label{tab:experiments-appendix}
\end{table*}

\subsection{Example of Optimal Play in {\em Avalon}}\label{sec:appendix-5p2m-eqm}

In this section, we fully describe one of the equilibria claimed in \Cref{tab:experiments}: the case of 5 players with Merlin and both \pmin players ``Mordred'' (hidden from Merlin). We choose this version because it is by far the easiest to describe the equilibrium: the other new values listed in \Cref{tab:experiments} are very highly mixed and difficult to explain.

Similarly to what happens in the example provided in \Cref{sec:example-intro}, this case is different from the case with no Merlin (even if Merlin does not have any useful knowledge), due to the effect of added correlation. In particular, if Merlin were not known to the mediator, then the equilibrium value would be $3/10 \cdot 2/3 = 2/10$: $3/10$ from the value of pure {\em Resistance} (no Merlin), and the factor of $2/3$ from a blind Merlin guess.

By the discussion above, it suffices to analyze the reduced game with three missions, each of size $3$, where \pmax need only win one mission to win the game.

The following strategy pair constitutes an equilibrium for this {\em Avalon} variant:

\paragraph{Strategy for \pmax} The mediator's strategy is different depending on how many players claim to it to be Merlin.

{\em One claim.} Randomly label the players A through E such that A is Merlin. Send missions ABC, ABD, and ABE.

There are six possible correct sets (ABC, ABD, ABE, ACD, ACE, ADE), and \pmax has three guesses, so \pmax wins the regular game with probability $1/2$. If the first mission succeeds, \pmin learns nothing about Merlin, so Merlin is guessed correctly with probability $1/3$. Otherwise, \pmin knows that Merlin is either A or B, so Merlin is guessed correctly with probability $1/2$. Thus, this strategy attains value $(1/6)(2/3) + (1/3) (1/2) = 5/18.$

{\em Two claims.} Randomly label the players A through E such that A and B are the Merlin claimers. Send missions ACD, BCE, and ADE.

{\em Three claims.} The non-claimants are guaranteed to be good. Pick a non-claimant at random, pretend that they are Merlin, and execute the strategy from the one-claim case. This does strictly better than the one-claim case, because \pmin learns only information about who is {\em not} Merlin instead of who {\em is} Merlin. (This strategy may not be subgame-perfect, but for the sake of this analysis it is enough for this strategy to achieve value {\em at least} $5/18$).

The analysis of the one-claim case applies nearly verbatim: there are six possible correct sets, and the second mission narrows down \pmin's list of possible Merlins from 3 to 2 (namely, A or D), so again the value for \pmin is $5/18$.

\paragraph{Strategy for \pmin}
Emulate the behavior of a non-Merlin \pmax-player (\ie, do not claim to be Merlin). 
\begin{enumerate}
    \item If the first guess by \pmax is correct, guess Merlin at random from the good players.
    \item If the second guess by \pmax is correct, guess Merlin at random from the good players included on both guesses.
    \item If the third guess by \pmax is correct, guess Merlin at random among all good players sent on at least two missions, weighting any player sent on all three missions double. (for example, if there was one player sent on all three missions and one player sent on two, guess the former with probability $2/3$ and the latter with probability $1/3$).
\end{enumerate} 
Against this strategy, the mediator's only decision is what three distinct missions to send. 

\begin{enumerate}
    \item If Merlin is sent on {\em one} mission, the probability of winning is at most $1/6$ because only that mission has a chance of being the correct one.
    \item If Merlin is sent on {\em two} missions, then  the probability of one mission passing is at most $1/3$. If the first mission passes, Merlin is guessed correctly with probability $1/3$. If the second mission passes, Merlin may not be guessed correctly at all (it is possible for Merlin to only have been sent on the second mission, but not the first). If the third mission passes, Merlin is guessed correctly with probability at least $1/5$. Thus, the probability of winning for good is at most $(1/3) (1-(1/3 + 0 + 1/5)/3) = 37/135 < 5/18.$
    \item If Merlin is sent on {\em all three} missions, the probability of a mission passing is at most $1/2$. If the first mission passes, Merlin is guessed correctly with probability $1/3$. Otherwise, Merlin is guessed correctly with probability at least $1/2$, since the only way to decrease this probability under $1/2$ would be for the last mission to pass {\em and} for another player to also have been sent on all three missions, but in that case it is impossible for anyone to have been sent on two missions. Therefore, the probability of winning for good is at most $(1/2) (1 - (1/3 + 1/2 + 1/2) / 3) = 5/18.$
\end{enumerate}

\end{document}